\documentclass[a4paper]{article}

\usepackage{authblk} 
\usepackage{amsfonts,amsmath,amssymb,mathtools,amsthm} 
\usepackage{enumerate} 
\usepackage[hidelinks]{hyperref} 
\usepackage{tikz-cd} 
\usepackage{relsize}
\usepackage{graphicx} 
\usepackage{cleveref} 
\usepackage{endnotes}

\let\footnote=\endnote

\usetikzlibrary{patterns,decorations.pathreplacing}
\tikzcdset{every label/.append style = {font = \small}} 

\definecolor{grey}{RGB}{100,100,100}

\newtheorem{theorem}{Theorem}[subsection]
\newtheorem{lemma}[theorem]{Lemma}

\newtheorem{definition}[theorem]{Definition}

\tikzset{
	labl/.style={rotate=90, inner sep=.5mm}
}

\newcommand{\Efrac}[2]{ 
  \mathchoice
    {\ooalign{%
      $\genfrac{}{}{1.2pt}0{#1}{#2}$\cr%
      $\color{white}\genfrac{}{}{.4pt}0{\phantom{#1}}{\phantom{#2}}$}}%
    {\ooalign{%
      $\genfrac{}{}{1.2pt}1{#1}{#2}$\cr%
      $\color{white}\genfrac{}{}{.4pt}1{\phantom{#1}}{\phantom{#2}}$}}%
    {\ooalign{%
      $\genfrac{}{}{1.2pt}2{#1}{#2}$\cr%
      $\color{white}\genfrac{}{}{.4pt}2{\phantom{#1}}{\phantom{#2}}$}}%
    {\ooalign{%
      $\genfrac{}{}{1.2pt}3{#1}{#2}$\cr%
      $\color{white}\genfrac{}{}{.4pt}3{\phantom{#1}}{\phantom{#2}}$}}%
}

\let\fun\lambda
\newcommand{\subst}[2]{[{#1}/{#2}]}

\newcommand{\I}{\mathbb{I}}
\newcommand{\U}{\mathcal{U}}
\newcommand{\zero}{\mathsf{0}}
\newcommand{\one}{\mathsf{1}}
\renewcommand{\path}[3]{\mathsf{path}_{#1}(#2, #3)}
\newcommand{\pathd}[3]{\mathsf{pathd}_{#1}(#2, #3)}
\newcommand{\refl}[1]{\mathsf{refl}_{#1}}
\newcommand*\sq{\;\mathbin{\vcenter{\hbox{\rule{.3ex}{.3ex}}}}\;}
\newcommand*\rsq{\; \mathbin{\vcenter{\hbox{\rule{.3ex}{.3ex}}}}_{\mathsf{r}} \;}
\newcommand*\lsq{\; \mathbin{\vcenter{\hbox{\rule{.3ex}{.3ex}}}}_{\mathsf{l}} \;}
\newcommand{\symm}[1]{{#1}^{-1}}
\newcommand{\trans}[2]{{#1}\sq{#2}}
\newcommand{\dfun}[3]{\prod_{(#1 : #2)}{#3}}

\newcommand{\pabs}[2]{\fun{#1}.{#2}}

\newcommand{\id}{\mathsf{Id}}
\newcommand{\filler}{\mathsf{fill}}

\newcommand{\lto}[2]{{#1} \rightsquigarrow {#2}}

\title{\bfseries{Naive cubical type theory}}

\author[1]{Bruno Bentzen}

\affil[1]{Carnegie Mellon University, \protect\\ Pittsburgh, PA, USA}
\affil[ ]{\texttt{bbentzen@andrew.cmu.edu}}

\date{}

\makeatletter
\newcommand\Author{Bruno Bentzen}
\let\Title\@title
\def\ps@mystyle{%
      \let\@oddfoot\@empty\let\@evenfoot\@empty
      \def\@oddhead{%
       \ifodd\value{page}\relax
          \hfill{Naive cubical type theory}\hfill\makebox[0pt][l]{\thepage}%
      \else
          \makebox[0pt][l]{\thepage}\hfill\Author\hfill%
      \fi%
      }%
      \let\@mkboth\markboth}
\makeatother
\begin{document}

\maketitle


\begin{abstract}
\noindent This paper proposes a way of doing type theory informally, assuming a cubical style of reasoning. It can thus be viewed as a first step toward a cubical alternative to the program of informalization of type theory carried out in the homotopy type theory book for dependent type theory augmented with axioms for univalence and higher inductive types. We adopt a cartesian cubical type theory proposed by Angiuli, Brunerie, Coquand, Favonia, Harper, and Licata as the implicit foundation, confining our presentation to elementary results such as function extensionality, the derivation of weak connections and path induction, the groupoid structure of types, and the Eckmman--Hilton duality.
\end{abstract}


\section{Introduction} \label{intro}

\noindent Cubical type theory is a flavor of higher-dimensional type theory that is more directly amenable to constructive interpretations~\cite{cchm,cartesian,afh}. It emerged as an alternative to the original formulation of homotopy type theory~\cite{hottbook}, which is based on an extension of dependent type theory with axioms asserting univalence and the existence of higher inductive types. 
Unlike this axiomatic formulation, which is well known to block computation since axioms postulate new canonical terms for some types without specifying how to compute with them, cubical type theory relies on cubical methods to properly specify the computational behavior of terms at higher dimensions. But, despite achieving a presentation of higher-dimensional type theory with computational content, the appeal to the sophisticated machinery of cubical reasoning makes cubical type theory even more impenetrable to the uninitiated. 
As of today, the only way to learn cubical type theory is either through the study of cubical models of type theory or using proof assistants that implement cubical type theory~\cite{redtt,vezzosi2019cubical}. 
This can make it hard to present cubical ideas to a wider mathematical audience.

In the spirit of Halmos' seminal book~\cite{halmos1974naive}, the naive type theory project introduced by Constable~\cite{constable2002naive} and recently advocated by Altenkirch~\cite{altenkirch2019naive} is aimed at making type theory more accessible to mathematicians by introducing the subject in an informal but, in principle, formalizable way, that is, by proposing an intuitive presentation independent of any technical appeal to the rules of inference of the formalism. In this sense, ``naive" is in contrast with ``formal", meaning that naive type theory can be well understood as formal type theory approached from a naive perspective~\cite{halmos1974naive,constable2002naive}. 
Ideally, it should be possible to make a naive type theory formally precise. This means that underlying a naive type theory one should always be able to find an implicit formal system made up of structural rules, rules of equality, type former rules like formation, introduction, elimination, and computation, and other rules defining a collection of types, their terms, and equality conditions. 
The point here is that we can almost always forget about those technicalities when assuming the naive point of view, according to which a type theory is viewed as an intuitive semantics rather than a body of axioms and rules of inference.  

For higher-dimensional type theories in particular, adopting such a naive approach amounts to developing an informal but rigorous explanation of type-theoretic reasoning at higher dimensions. This is an important enterprise that can be expected to provide at least two major benefits: 

\begin{enumerate}[(1)]
	\item If higher-dimensional type theory is to be taken seriously as a foundation for mathematics or research program, then, it should be accessible with a minimum of logical formalism to nonexperts. 
	
	\item Pen-and-paper proofs given in a homotopically-inspired informal language for mathematics could be more closely related to practices of working mathematicians such as the identification of isomorphic structures. 
	Moreover, proof mechanization might require significantly less effort, as type theory is the basis of several proof-assistants. 
\end{enumerate}

For homotopy type theory, as originally formulated, this informalization was accomplished in the canonical book on the subject~\cite{hottbook}, in which the theory is systematically developed from scratch with the use of language and notation that are similar to that of ordinary mathematics, but without making precise reference to the axioms or rules of inference that establish the formal system. 
Put differently, the so-called homotopy type theory book~\cite{hottbook} develops a naive type theory, a rigorous style of doing everyday mathematics informally assuming homotopy type theory as the underlying foundation, which is then used to give informal proofs of theorems from various areas of mathematics, such as logic, set theory, category theory, and homotopy theory throughout the book. 

This paper can be viewed as an initial effort to introduce readers to the naive way of doing cubical type theory in an analogous manner. 
Thus, the main goal here is to propose a naive presentation for cubical type theory that has a comparable degree of rigor, while also highlighting the distinctive aspects of the cubical approach to higher-dimensional type theory by means of proofs of some elementary results. 
For that reason, the reader may notice that, in this paper, arguments by the principle of path induction, which is used in homotopy type theory as the elimination rule for the path type~\cite[\S1.12]{hottbook}, are almost completely ignored in favor of purely cubical arguments. 
This often results in diagrammatic proofs that are conceptually simpler, cubically speaking, but may require more ingenuity. 
Now, considering that cubical type theory comes in many different forms, depending on the structure one imposes on the cube category the framework is built on, we should emphasize that our naive presentation is founded on cartesian cubical type theory~\cite{cartesian}, a formal theory developed by Angiuli, Brunerie, Coquand, Favonia, Harper, and Licata based on the cartesian cube category, that is, the free category with finite products on an interval object~\cite{awodey2018cubical}, a category of cubes which has faces, degeneracies, symmetries, diagonals, but no connections or reversals\textemdash unlike the De Morgan cube category and the cubical type theories based on this structure~\cite{cchm,chm}. 

The remainder of this paper is organized as follows: in the next section, a short introduction to cubical type theory from the naive point of view is given and its semantics is described informally. Then, Kan operations are discussed in some detail. Next, we derive weak connections and give informal proofs of the groupoid structure of types. Subsequently, we consider dependent paths and heterogeneous composition. Finally, proofs of path induction and the Eckmann--Hilton duality are presented, followed by a section containing some final remarks. 
This paper is intended to be self-contained, but we assume throughout some general familiarity with the concepts of homotopy type theory. 

\section{The cubical point of view} \label{cctt}

Naive cubical type theory is the idea that cubes are the basic shapes used to characterize the structure of higher-dimensional mathematical objects. 
It is, however, grounded on the same homotopical intuition~\cite{awodey2009homotopy,voevodsky2006very} which regards a type $A$ as a space, a term $a : A$ as a point of the space $A$, a function $f : A \to B$ as a continuous map, a path $p : \path{A}{a}{b}$ as a path from point $a$ to point $b$ in the space $A$, a universe type $\U$ as a space, the points of which are spaces, and a type family $P : A \to \U$ as a fibration. 
Like homotopy type theory, spaces are understood purely from the point of view of homotopy theory, meaning that homotopy equivalent spaces are equal up to a path.

The distinctive cubical perspective of type theory starts with the consideration of an abstraction of the unit interval in the real line, that is, a space consisting of only two points, $\zero$ and $\one$, which we call the \textit{interval type} $\I$. 
Thus, an interval variable $i : \I$ is a type-theoretic abstraction of the idea of a point varying continuously in the unit interval. 
Traditionally, a path in a topological space is a continuous function from the unit interval. 
This point-set topological description is generalized in cubical type theory in the sense that we represent paths as functions out of the interval. As will be discussed in the remainder of this section, this change of perspective regarding paths is the key ingredient that allows us to visualize them as abstract cubes in higher dimensions. 

\subsection{The type of paths} \label{path}

It is useful to have a type of paths. Certainly, the most obvious way of obtaining such a type is to use the function type itself. 
We shall often refer to the type of functions from the interval $\I \to A$ as the \textit{line type}, and call its terms lines. It is convenient to work with paths with arbitrary endpoints using line types, but as soon as we start considering paths from a particular point to another point, the limitations of this strategy begin to become apparent. 
In order to better deal with paths with fixed endpoints, we need a slightly more sophisticated variant of the function type called the \textit{path type}, an internalization of functions from the interval that makes their endpoints fully explicit. Hence given any type $A : \U$ and terms $a, b : A$, we can form the type of paths starting from $a$ to $b$ in $A$, which we denote by $\path{A}{a}{b}$ and call the path type of $A$ and $a$ and $b$.\footnote{
	In homotopy type theory, the type of paths is usually referred to as the identity type and its inhabitants are treated as witnesses that two objects are equal. We will avoid this terminology here given that it is open to question whether the highly technical conception of a path in higher-dimensional type theory is indeed capable of reflecting ordinary mathematical equality~\cite{ladymanpresnell2019universes,bentzen2020different}. Because the identity type is inductively defined in homotopy type theory, we shall sometimes refer to it as the ``inductive path type'' throughout the text. } 

What can we do with paths in the path type? Given a path $p : \path{A}{a}{b}$ and an interval variable $i : \I$, we can apply the path to obtain a term of the type $A$ depending on $i$, denoted $p (i)$ and called the value of $p$ at position $i$. Moreover, applying a path to one of the interval endpoints should result in the corresponding endpoint of the path. This means that for every term $p : \path{A}{a}{b}$ the following endpoint equalities hold definitionally: 

$$p (\zero) \equiv a : A \quad\text{and}\quad p (\one) \equiv b : A.$$

\noindent The canonical way to construct paths is by abstraction: given a term $a : A$ which may depend on $i : \I$, we write $\pabs{i}{a}$ to indicate a path from $a \subst{\zero}{i}$ to $a \subst{\one}{i}$ in $A$, where, by convention, the scope of the binding $\pabs{-}{}$ extends over the entire expression to its right unless otherwise noted by the use of parentheses. 

It should come as no surprise that paths require equalities similar to those common for functions in the lambda calculus. 
More specifically, when a path abstraction term is applied to an interval point, we require a computation that plays the role of $\beta$-reduction: 
$$(\pabs{i}{a})(j) \equiv a \subst{j}{i}.$$

\noindent We also expect that ``every path is a path abstraction", the uniqueness principle for the path type, meaning that we also consider the following $\eta$-expansion rule: 

$$\qquad p \equiv \pabs{i}{(p(i))}  \qquad \text{(when \textit{i} does not occur in \textit{p}).} $$
This identifies any path $p$ with the ``path that applies $p$ on the interval", thus endowing paths with an extensional aspect.
The use of this equality is crucial to the derivation of path induction in \S\ref{pathinduction}.

\subsection{How we should think of paths} 

The use of diagrams will play a key role in our visualization of paths, for many cubical type-theoretic problems can be rigorously stated and solved using diagrammatic arguments. How should we think of paths cubically? This view follows naturally from the conception of paths as functions from the interval and the homotopical intuition. First of all, we visualize a term $a : A$ as a point: 

$$ \cdot \; a $$ 

\noindent In the first dimension we think of a term $p : \I \to A$ as a path from the point $p(\zero)$ to the point $p(\one)$. Given an interval variable $i : \I$, such a path can be visualized as an abstract line $p(i)$ in the ``direction" $i$ that goes from the initial to the terminal point of the path, as shown in the following diagram: 

\vspace{2mm}

\begin{tikzcd}
	&& p(\zero) \arrow[rrrrrr,"p(i)"] &&&&&& p(\one) && \arrow[r,shorten >= 12pt,-latex,swap,start anchor=center,"i" near start] & {} 
\end{tikzcd}

\vspace{2mm}

\noindent Extending the interpretation to higher dimensions, it is natural to think of a function $h : \I \to \I \to A$ as a homotopy of paths. 
Such a homotopy consists of two simultaneous continuous deformations and it can be visualized as an abstract square having four paths at each edges, as shown in the following diagram, where the lines $h(\zero)$ and $h(\one)$ are the bottom and top faces in the direction $i$, and $\fun j.h(j, \zero)$ and $\fun j.h(j, \one)$ the left and right faces in the $j$-direction, respectively: 

\vspace{2mm}

\begin{tikzcd}
	&& h(\one, \zero) \arrow[rrrrr,"h(\one)"] \arrow[ddrrrrr, phantom,"h"] &&&&& h(\one, \one) && {} \\ 
	&&&&&&&&& \arrow[r,shorten >= 12pt,-latex,swap,start anchor=center,"i" near start] \arrow[u,shorten >= 5pt,-latex,start anchor=center,"j" near start] & {} \\
	&& h(\zero, \zero) \arrow[rrrrr, swap,"h(\zero)"] \arrow[uu,"{\fun j. h(j, \zero)}"] &&&&& h(\zero,\one) \arrow[uu,swap,"{\fun j. h(j, \one)}"]
\end{tikzcd}

\vspace{2mm}

\noindent In the third dimension, we consider homotopies between homotopies, which, as expected can be pictured as cubes. It is hard enough to visualize paths at even higher dimensions, but, most certainly, the reader can already guess the general pattern here: at dimension $n+1$, we have $n+1$-dimensional paths, which can be visualized as hypercubes with $2{(n+1)}$ faces formed by $n$-dimensional cubes. Finally, it is worth noting that this cubical structure applies for types $A : \U$ as well, since types are just terms of a type of universes. 

\subsection{How can we use paths?}

Already at this stage we can see how this novel account of paths as functions out of the interval is a significant departure from the traditional approach taken in homotopy type theory~\cite{hottbook}, in which the path type is defined as an inductive type generated by a single reflexivity constructor. 
The following lemmas will serve to show some fundamental differences of the cubical approach. 

We start with the fact that a term analogous to this reflexivity constructor is definable by considering constant functions from the interval, thus providing us a way to trivially regard terms as degenerate paths. 


\begin{lemma}\label{def:refl}
	For every type $A$ and every $a : A$, there exists a path
	$$\path{A}{a}{a}$$
	called the reflexivity path of $a$ and denoted $\refl{a}$.
\end{lemma}
\begin{proof} 
	Suppose that $i : \I$ is a fresh interval variable. By assumption, $a$ does not depend on $i$, meaning that $\pabs{i}{a}$ gives us a path starting from $a$ to $a$ in $A$. 
\end{proof}

Since in the homotopy interpretation we view functions as continuous maps, it is natural to expect that functions preserve paths, as shown in the homotopy type theory book \cite[\S2.2]{hottbook}. Here we state this property with the following lemma for non-dependent functions:

\begin{lemma}[Function application~\cite{bch14,cchm}]
	Given a function $f : A \to B$ and terms $a, b : A$, we have an operation
	$$\mathsf{ap}_f : (\path{A}{a}{b}) \to (\path{B}{f(a)}{f(b)})$$
	such that $\mathsf{ap}_f (\refl{a}) \equiv \refl{f(a)}$.
\end{lemma}
\begin{proof}
	Given $p : \path{A}{a}{b}$ and an interval variable $i : \I$, we have $f(p (i)) : B$, for we can apply $f : A \to B$ to $p(i) : A$. By abstraction, we have a path
	$$\mathsf{ap}_f (p) :\equiv \pabs{i}{f(p (i))}.$$
	from $f (a)$ to $f (b)$ in $B$, since $p$ is a path from $a$ to $b$. Moreover, we have 
	
	\begin{align*}
	\mathsf{ap}_f (\refl{a}) &\equiv \pabs{i}{f(\refl{a} (i))} \\
	&\equiv \pabs{i}{f(a)} \\
	&\equiv \refl{f(a)}
	\end{align*}
	as requested. 
\end{proof}

We note that $\mathsf{ap}$ behaves strictly functorially in the sense that the following are definitional equalities for identity functions, compositions of functions, and constant functions~\cite{bch14,cchm}:

\begin{align*}
\mathsf{ap}_{id_A} (p) &\equiv p \\
\mathsf{ap}_{f \circ g} (p) &\equiv \mathsf{ap}_{f} (\mathsf{ap}_{g} (p)) \\
\mathsf{ap}_{\fun \_.a} (p) &\equiv \refl{a}.
\end{align*}
In homotopy type theory~\cite{hottbook}, in which the path type is defined as an inductive family, those equalities only hold up to homotopy. 

Another crucial difference is that, with the cubical notion of path, we are able to prove that two pointwise equal functions are equal up to a path. This is known as function extensionality, a property that cannot be obtained without the use of axioms in homotopy type theory~\cite[\S2.9,\S4.9]{hottbook}. 
This lack of function extensionality is a common occurrence shared by many intensional forms of dependent type theory (except for observational type theory). 

\begin{theorem}[Function extensionality~\cite{cchm}] \label{funext}
	Suppose that $f, g : \prod_{(x : A)} B(x)$ are functions. There is an operation
	$$\mathsf{funext}_{f,g} : \prod_{(x : A)} \path{B(x)}{f(x)}{g(x)} \to \path{\prod_{(x : A)} B(x)}{f}{g}$$
\end{theorem}
\begin{proof}
	If we are given $H : \prod_{(x : A)} \path{B(x)}{f(x)}{g(x)}$ and an arbitrary $x : A$, we have a path $H(x)$ from $f(x)$ to $g(x)$ in $B(x)$. For $i : \I$, $H(x)$ at $i$ gives a path from $\fun x.f(x)$ to $\fun x.g(x)$ in $A \to B$,
	$$\fun i.\fun x. (H(x) (i)) : \prod_{(x : A)} B(x).$$
	By $\eta$-conversion, it can be seen that this path goes from $f$ to $g$. 
\end{proof}

\section{There are enough paths} \label{hd-op}

So far we have only discussed general properties of paths. However, nothing about what has been said leaves us in a position to transport terms along paths or invert and concatenate paths, to give some examples. Even though compared to homotopy type theory we get closer to the types as spaces intuition in the cubical setting, since paths are viewed as functions out of the interval, this attitude alone cannot provide the required structure to support it. 
There is even a certain advantage in homotopy type theory because path induction is accepted as a primitive form of reasoning and all the structure needed to model spaces can be derived from it~\cite[\S2]{hottbook}. 
In the end path induction guarantees that everything in homotopy type theory can be interpreted in a model of spaces such as that of Kan simplicial sets, and there is nothing strictly simplicial about the spatial intuition required to do homotopy type theory informally. 

In cubical type theory, on the other hand, the connection between the theory and the intended model of spaces happens at a much higher level, except for the fact that Kan cubical sets are considered instead~\cite{cchm,cartesian,afh}. So when doing naive cubical type theory it is no longer possible to maintain the same superficial level of intuition because cubical reasoning becomes imperative. 
In order to properly model spaces in a cubical setting it is therefore unacceptable to rely on this generic homotopy type-theoretic strategy of using path induction as a fundamental reasoning principle. 
Instead, to ensure that types have enough structure we impose cubically-inspired primitive operations (Kan operations) such as \textit{transport} and \textit{composition}~\cite{cartesian,afh}. 
They stipulate certain requirements that paths in a type should satisfy, and, to ensure that all terms compute properly, every type must come with its own specific operations, which determine how we should understand a Kan operation on a type in terms of reductions to their constructors or eliminators~\cite{cartesian}. 
The general aspects of transport and composition are discussed in the remainder of this section.

\subsection{Transportation along paths} 

Transport is a cubical generalization of the transport lemma~\cite[Lem 2.3.1]{hottbook}. 
Recall that, due to univalence~\cite{voevodsky2009notes}, a principle that characterizes the path space of the universe type~$\U$, a path between types is (equivalent to) a homotopy equivalence between spaces~\cite{hottbook}. 
In the cubical setting, transport states that, given any path between types $A: \I \to \U$ and any term $a : A(i)$, we have a term of the type $A (j)$, called the transport of $a$ from $i$ to $j$ along $A$, and denoted by $a^{\lto{i}{ j}}_{A} : A(j)$. We also require that static transportations have no effect, which means that, when $i$ is $j$, we have $a^{\lto{i}{i}}_{A} \equiv a$.

In general, transporting along a degenerate path is not strictly the same as doing nothing, meaning that 
$a^{\lto{i}{ j}}_{\refl{A}} \not\equiv a.$ 
This is in contrast with the transport operation of homotopy type theory~\cite{hottbook}, in which transportation along reflexivity is taken to be the identity function. The following lemma shows in what sense a generalization of this operation is taking place: 

\label{coe}

\begin{lemma}\label{def:transport}
	Given a type family $C : A \to \U$, terms $a, b : A$ and a path $p : \path{A}{a}{b}$, there is a function
	$$p_* : C(a) \to C(b).$$
\end{lemma}
\begin{proof} 
	Assume we have $c : C(a)$. Note first that we have a path of types 
	
	$$D(p) :\equiv \fun i.C(p(i)) : \I \to \U$$
	from $C(a)$ to $C(b)$. The transportation of $c$ from $\zero$ to $\one$ along $D(p)$ gives us
	
	$$c^{\lto{\zero}{\one}}_{D(p)} : C(b)$$
	as desired.
\end{proof}

Although $p_*(\refl{a})$ is not the identity function up to definitional equality, this can be shown to be the case up to a path. 

\begin{lemma}\label{def:transport}
	For a type family $C : A \to \U$ and $a : A$, we have a path
	$$\path{C(a) \to C(a)}{(\refl{a})_*}{\id_{C(a)}}.$$
\end{lemma}
\begin{proof}
	The idea is to use function extensionality. By Theorem~\ref{funext}, we may assume that $c : C(a)$. It remains to find a path from $(\refl{a})_*(c)$ to $c$. But by definition, $(\refl{a})_*(c) \equiv c^{\lto{\zero}{\one}}_{\fun \_.C(a)}$. 
	Now fix $i : \I$. We observe that $c^{\lto{i}{\one}}_{\fun \_.C(a)} : C(a)$, since $(\fun \_.C(a))(j) \equiv C(a)$ for any $j : \I$. 
	But recall that the static transportation $c^{\lto{\one}{\one}}_{\fun \_.C(a)}$ has no effect, so this term is definitionally equal to $c$. 
	In other words, $\fun i. c^{\lto{i}{\one}}_{\fun \_.C(a)}$ is the path we are looking for. 
\end{proof}

\subsection{Composition of paths} \label{hop}

Composition ensures that every open cube in a type can be filled. For example, at dimension one, if we are given adjacent lines $p , q, r : \I \to A$ such that the initial point of $p$ matches the initial point of $q$ and the terminal point of $p$ matches the initial point of $r$, composition provides a new line, the composite, that goes from the terminal point of $q$ to the terminal point of $r$.
	
\label{hcom}

\vspace{2mm}

\begin{tikzcd}
	& q(1) \arrow[rrrrr,dotted,""] \arrow[ddrrrrr, phantom,""] &&&&& r(1) && {} \\ 
	&&&&&&&& \arrow[r,shorten >= 12pt,-latex,swap,start anchor=center,"i" near start] \arrow[u,shorten >= 5pt,-latex,start anchor=center,"j" near start] & {} \\
	& q(0) \equiv p(0) \arrow[uu,"q(j)"] \arrow[rrrrr,swap,"p(i)"] &&&&& p(1) \equiv r(0) \arrow[uu,swap,"r(j)"]
\end{tikzcd}

\vspace{2mm}

\noindent In fact, composition states more than that. It asserts the existence of the whole square witnessing the filling of the open box (the square in the diagram above). This is called the \textit{filler} of the composition. 
When $j$ is $\one$, the filler gives us the composite, the missing face of the open box. 
The remaining faces of the filler are the faces of the open box. So, when $j$ is $\zero$, the filler equals $p$ (which we may call the cap of the composition) and when $i$ is $\zero$ or $\one$, it will be $q$ (the left tube) or $r$ (the right tube), respectively. More generally, we refer to any open shape satisfying the relevant adjacency conditions as a \textit{composition scenario}. 

To illustrate the use of composition, we shall now prove symmetry and transitivity for paths with the definition of inverse and concatenation functions. Let us consider the former first. 
In~cubical type theories based on a De Morgan structure~\cite{cchm,chm}, path inversion comes for free with a primitive (strict) reversal operator but in a cartesian account some work needs to be done to derive it from the Kan operations. 


\begin{lemma}\label{lem:inv}
	For every type $A$ and every $a, b : A$, there is a function
	$$\path{A}{a}{b} \to \path{A}{b}{a}$$
	called the inverse function and denoted $p \mapsto \symm{p}$.
\end{lemma}
\begin{proof}
	Suppose that $i, j : \I$. The idea of the following proof is to observe that $p : \path{A}{a}{b}$ gives a ``line" $p (j) : A$ in the $j$-direction from $a$ to $b$ in $A$. Since its initial point is $a$, we have an open box whose faces are formed by the lines $p(j)$ (left tube), $\refl{a}(i)$ (cap) and $\refl{a}(j)$ (right tube).
	Note that the latter two lines are degenerate, i.e. by definition, they are the same as $a$. Degeneracy is indicated using double bars in the following diagram:
	
	
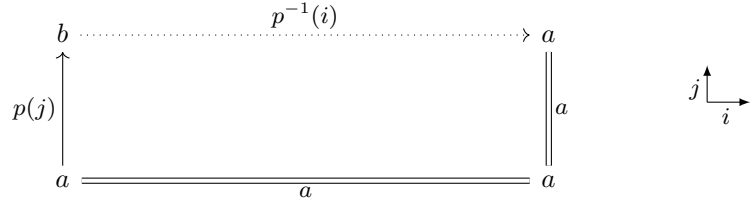
\vspace{2mm}
	
	\begin{tikzcd}
		&& b \arrow[rrrrrrr,dotted,"{\symm{p} (i)}"] \arrow[ddrrrrrrr, phantom,""] &&&&&&& a && {} \\ 
		&&&&&&&&&&& \arrow[r,shorten >= 12pt,-latex,swap,start anchor=center,"i" near start] \arrow[u,shorten >= 5pt,-latex,start anchor=center,"j" near start] & {} \\
		&& a \arrow[rrrrrrr,swap,equal,"a"] \arrow[uu,"p (j)"] &&&&&&& a \arrow[uu,swap,equal,"a"]
	\end{tikzcd}
	
	\vspace{2mm} 
	
	\noindent By composition, this open box must have a lid (top), an $i$-line from $b$ to $a$ in $A$, which gives us, by path abstraction, the required path $\symm{p}$.
\end{proof}


Fillers of compositions will be relevant to our constructions later on, so it can be convenient to have a special symbolism to talk about them directly. That is the purpose of the following notation: 

\begin{definition}
	Given a composition scenario, $i, j : \I$, and a composite $p$ for an open cube in the $i$ direction, then $\filler_{j}(p)$ stands for its filler in the $j$ direction. 
\end{definition}

When no confusion occurs, we may also write it $\filler(p)$. Thus, the $(i,j)$-square defined in the proof of Lemma~\ref{lem:inv} can be denoted by $\filler(\symm{p}(i))$. 

The next lemma defines path concatenation:


\begin{lemma}\label{lem:comp}
	For every type $A$ and every $a, b, c : A$, there is a function
	$$\path{A}{a}{b} \to \path{A}{b}{c} \to \path{A}{a}{c}$$
	denoted $ p \mapsto q \mapsto p \sq q$. We call $p \sq q$ the concatenation of $p$ and $q$.
\end{lemma}
\begin{proof} 
	Given paths $p : \path{A}{a}{b}$ and $q : \path{A}{b}{c}$, we can construct an $i$-line $p(i)$ from $a$ to $b$ and a $j$-line $q(j)$ from $b$ to $c$. Once again, we note that we have an open square, 
	
	
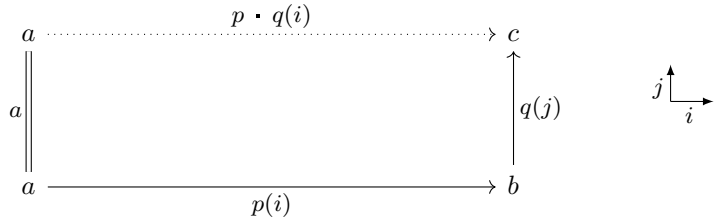
\vspace{2mm}
	
	\begin{tikzcd}
		&& a \arrow[rrrrrrr,dotted,"\trans{p}{q} (i)"] \arrow[ddrrrrrrr, phantom,""] &&&&&&& c && {} \\ 
		&&&&&&&&&&& \arrow[r,shorten >= 12pt,-latex,swap,start anchor=center,"i" near start] \arrow[u,shorten >= 5pt,-latex,start anchor=center,"j" near start] & {} \\
		&& a \arrow[rrrrrrr, swap,"p(i)"] \arrow[uu,"a",equal] &&&&&&& b \arrow[uu,swap,"q (j)"]
	\end{tikzcd}
	
	\vspace{2mm}
	
	\noindent 
	and the path $\trans{p}{q}$ is obtained by path abstraction on its composite. 
\end{proof}


We shall make diagrams as explicit as possible throughout the remainder of this paper, but, for the sake of readability, we omit labels for degenerate lines, since the reader should be able to correctly guess the information by checking the endpoints of the line. 
It is also important to bear in mind that, while we only considered the simplest composition scenario where we compose lines to form squares in the examples above, we can also compose squares to form cubes and, more generally, $n$-cubes to form $(n+1)$-cubes. We shall deal with compositions at dimension two in \S\ref{weakconnections}. Composition scenarios at even higher dimensions will not be discussed in this paper. 

In addition to the specification of the cap and tubes of an open cube, we also have to admit the specification of diagonals to make composition work properly in a cartesian setting, thereby allowing the diagonal of the filler to be definitionally equal to the designated diagonal \cite{cartesian,afh}. 
But, as we will not be using diagonals, we do not need to worry about this here. 

\subsection{The interval is not Kan} 

We suggested above that every type is Kan. In fact, the interval is the only exception to this rule, since we have been implicitly treating it as a ``type", but it actually does not support any Kan operations. Consider, for example, the identity path on the interval, which goes from $\zero$ to $\one$,

$$\fun i.i : \path{\I}{\zero}{\one}.$$ 

\noindent If the interval were Kan, then the identity path would have an inverse. But what could that be in our cartesian setting?  
Given that the interval lacks some properties required by typehood in a cubical type theory, it is often referred to as a pretype. To deal with this fact, we adopt the convention that the interval can only occur as the antecedent of a function type. 

\section{Two-dimensional constructions} \label{sec:twodim}

Now that we have defined path symmetry and concatenation we can show that they satisfy the groupoid laws up to homotopy, which means that we shall be mainly concerned with paths of one dimension higher than the ones we are given. This is the aim of this section. Two-dimensional paths are determined by squares, which represent a mutual identification of lines on the opposing sides, and, as a result, their construction often requires two-extent compositions.

\subsection{Weak connections} \label{weakconnections}

In cubical type theory, it is useful to have extra kinds of degeneracies known as connections, which can be thought of as meets and joins on the interval. 
Just as reversals, connections are built-in in a De Morgan setting~\cite{cchm,chm}. They are not hard to derive in cartesian cubical type theory, but some of the computation rules only hold up to a path (hence we call them ``weak" connections). 

\begin{lemma}[Meet] \label{lem:meet}
	For every type $A$ and $a, b : A$, there is an operation
	$$- (- \land -) : \path{A}{a}{b} \to \I \to \I \to A$$
	such that, for any $p : \path{A}{a}{b}$ and any $i, j : \I$, the following holds: 
	
	\vspace{2mm}
	
	\begin{tikzcd}
		&& a \arrow[rrrrrrr,"p (i)"] \arrow[ddrrrrrrr,phantom,"p (i \land j)"] &&&&&&& b && {} \\ 
		&&&&&&&&&&& \arrow[r,shorten >= 12pt,-latex,swap,start anchor=center,"i" near start] \arrow[u,shorten >= 5pt,-latex,start anchor=center,"j" near start] & {} \\
		&& a \arrow[rrrrrrr,swap,equal] \arrow[uu,equal] &&&&&&& a \arrow[uu,swap,"p (j)"]
	\end{tikzcd}
\end{lemma}

The proof of this lemma requires a two-extent composition and thus a few remarks are in order here before we proceed. In the one-dimensional case, as we have seen as far, we just have a pair of tubes, so an open square is enough to complete the composition. With two pairs of tubes, each pair corresponding to the dimension in question, we are required to form an open object of the next higher shape, i.e. a cube. 
In other words, assuming that we want to fill an open $(i, j, k)$-cube in the $j$ direction, we are expected to determine its $(i, k)$-face (bottom), two $(k, j)$-faces (left and right) and two $(i, j)$-faces (back and front), and those squares must all be adjacent up to definitional equality. If this holds, then the composite is a $(i, k)$-square. (For the same reasons, performing an $n$-extent composition requires a construction of an $(n+1)$-cube).

Finally, it is often convenient to illustrate two-extent compositions using a two-dimensional structure in the form of a cube seen from above.\footnote{
	The author wishes to thank Lo\"ic Pujet for this suggestion.} 
The following diagram indicates how we shall be drawing two-extent compositions in the paper:  

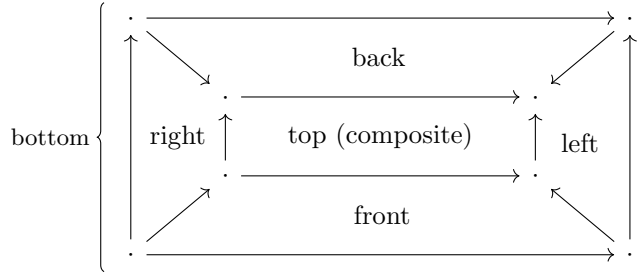
\begin{center}
	\begin{tikzcd}
		& \cdot \arrow[ddd, start anchor=north, end anchor=south, no head, xshift=-1em, decorate, decoration={brace,mirror}, "\text{bottom}" left=3pt] \arrow[dr]  \arrow[rrrrrr] & {} \arrow[drrrr,phantom,"\text{back}"] & {} &&&& \cdot \arrow[dl] \arrow[dddl,phantom,"\text{left}"]  \\
		&& \cdot \arrow[rrrr] \arrow[drrrr,phantom,"\text{top (composite)}"] &&&& \cdot     \\
		&& \cdot \arrow[rrrr,swap] \arrow[u] &&&& \cdot \arrow[u] & {} & {} \\
		& \cdot \arrow[ur,swap] \arrow[uuu] \arrow[rrrrrr] \arrow[uuur,phantom,"\text{right}"] \arrow[rrrrrru,phantom,"\text{front}"] &&&&& {} & \cdot
		\arrow[ul] \arrow[uuu]  \\
	\end{tikzcd}
\end{center}

\noindent When referring to such diagrams we may call the center and outer squares the top (composite) and the bottom faces of the cube (filler), the top and bottom squares the back and front faces, and the left and right squares the left and right faces, respectively. 

\vspace{2mm}

\begin{proof}
	Given $p : \path{A}{a}{b}$, we are to find a $(i,j)$-square whose top face is $p (i)$, right face is $p (j)$, left and bottom faces are $a$. 
	First, by composition, we obtain, for any $i, j : \I$, a square that looks like a ``halfway" connection, called $p(i \land^{*}j)$. 
	
	
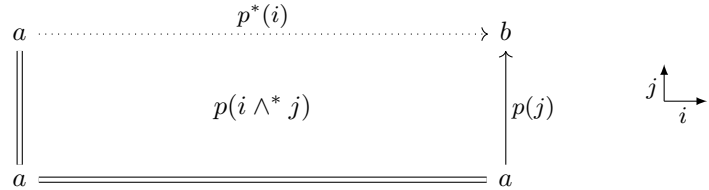
\vspace{2mm}
	\begin{tikzcd}
		&& a \arrow[rrrrrrr,dotted,"p^{*} (i)"] \arrow[ddrrrrrrr,phantom,"p(i \land^{*} j)"] &&&&&&& b && {} \\ 
		&&&&&&&&&&& \arrow[r,shorten >= 12pt,-latex,swap,start anchor=center,"i" near start] \arrow[u,shorten >= 5pt,-latex,start anchor=center,"j" near start] & {} \\
		&& a \arrow[rrrrrrr,equal,swap] \arrow[uu,equal] &&&&&&& a \arrow[uu,swap,"p (j)"]
	\end{tikzcd}
	\vspace{2mm}
	
	\noindent 
	Note that this square is not the one that we are looking for. Its top face is given by a certain path obtained by composition, which we denote $p^*$, and this path needs not be definitionally equal to $p$. 
	
	But we are able to fix this mismatch in a higher dimension by attaching this square to the top corner of an otherwise degenerate open cube in such a way that $p(- \land^{*} -)$ forms the back and right faces and the back right edge is $p^*$. This is depicted in the diagram below, where the bottom, front, and left faces of the open cube are $a$. We complete the proof with a two-extent composition, whose composite is shown as the shaded face in the diagram:
	
	
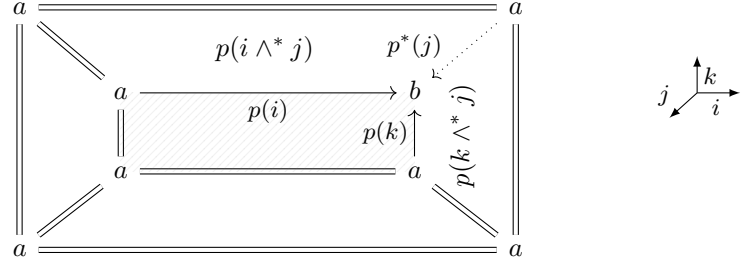
\vspace{2mm}
	
	\begin{tikzcd}
	[execute at end picture={
		\foreach \Valor/\Nombre in
		{
			tikz@f@9-2-4/a,tikz@f@9-2-8/b,tikz@f@9-3-4/c,tikz@f@9-3-8/d%
		}
		{
			\coordinate (\Nombre) at (\Valor);
		}
		\fill[pattern=north east lines,pattern color=grey,opacity=0.3]
		(b) -- (a) -- (c) -- (d) -- cycle;
	}]
		&& a \arrow[dr,equal]  \arrow[rrrrrr,equal] & {} \arrow[drrrr,phantom,"p(i\land^{*}j)"] &&&&& a \arrow[dl,swap,dotted,"p^* (j)" near end] \arrow[dddl,phantom,"p(k\land^{*}j)" labl] && {} \\
		&&& a \arrow[rrrr,swap,"p (i)"] &&&& b &&& \arrow[r,shorten >= 12pt,-latex,swap,start anchor=center,"i" near start] \arrow[u,shorten >= 12pt,-latex,swap,start anchor=center,"k" near start] \arrow[dl,shorten >= 28pt,-latex,swap,start anchor=center,"j" near start] & {}    \\
		&&& a \arrow[rrrr,swap,equal] \arrow[u,equal] &&&& a \arrow[u,"p (k)"] && {} \\
		&& a \arrow[ur,swap,equal] \arrow[uuu,equal] \arrow[rrrrrr,swap,equal] &&&&& {} & a
		\arrow[ul,equal, near end] \arrow[uuu,equal]  \\
	\end{tikzcd}

	\end{proof}

Before moving on to the next connection, it is worth pointing out that the connections derived here could be slightly improved by attaching squares to the diagonal face of the open cubes of their compositions, thereby making the diagonal of the connections definitionally equal to $p$.
	
\begin{lemma}[Join]\label{lem:join}
	Given a type $A$ and $a, b : A$, we have an operation:
	$$- (- \lor -) : \path{A}{a}{b} \to \I \to \I \to A$$
	such that, for $p : \path{A}{a}{b}$ and $i, j : \I$, we have:
	
	
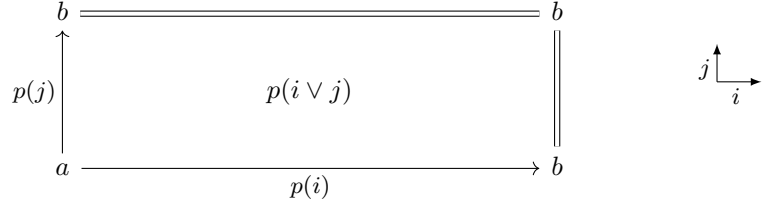
\vspace{2mm}
	
	\begin{tikzcd}
		&& b \arrow[rrrrrrr,equal] \arrow[ddrrrrrrr,phantom,"p (i \lor j)"] &&&&&&& b && {} \\ 
		&&&&&&&&&&& \arrow[r,shorten >= 12pt,-latex,swap,start anchor=center,"i" near start] \arrow[u,shorten >= 5pt,-latex,start anchor=center,"j" near start] & {} \\
		&& a \arrow[rrrrrrr,swap,"p (i)"] \arrow[uu,"p (j)"] &&&&&&& b \arrow[uu,swap,equal]
	\end{tikzcd}
\end{lemma}
\begin{proof} 
	Using the halfway meet connection constructed as part of the previous proof, we perform a two-extent composition on an open cube given by halfway meets (front and left), degenerate squares formed from lines (back and right) and points (bottom). The composite gives us the desired square:
	
	
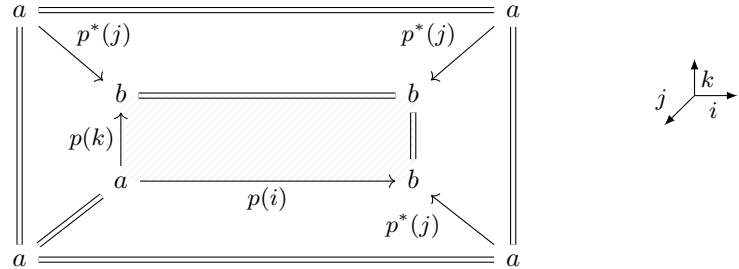
\vspace{2mm}
	\begin{tikzcd}
	[execute at end picture={
		\foreach \Valor/\Nombre in
		{
			tikz@f@11-2-4/a,tikz@f@11-2-8/b,tikz@f@11-3-4/c,tikz@f@11-3-8/d%
		}
		{
			\coordinate (\Nombre) at (\Valor);
		}
		\fill[pattern=north east lines,pattern color=grey,opacity=0.3]
		(b) -- (a) -- (c) -- (d) -- cycle;
	}]
		&& a \arrow[dr,"p^* (j)"] \arrow[rrrrrr,equal] &&&&&& a \arrow[dl,swap,"p^* (j)"] && {} \\
		&&& b \arrow[rrrr,equal] &&&& b &&& \arrow[r,shorten >= 12pt,-latex,swap,start anchor=center,"i" near start] \arrow[u,shorten >= 12pt,-latex,swap,start anchor=center,"k" near start] \arrow[dl,shorten >= 28pt,-latex,swap,start anchor=center,"j" near start] & {}    \\
		&&& a \arrow[rrrr,swap,"p (i)"] \arrow[u,"p (k)"] &&&& b \arrow[u,swap,equal] && {} \\
		&& a \arrow[ur,swap,equal] \arrow[uuu,equal] \arrow[rrrrrr,swap,equal] &&&&&& a
		\arrow[ul,"p^* (j)" near end] \arrow[uuu,equal]  \\
	\end{tikzcd}

	 To check that this is a well-formed composition, note that we set $a$ as the bottom $(i, k)$-face, $p(k \land^* j)$ as the left and $p^*(j)$ as the right $(k, j)$-faces, $p^* (j)$ as the back and $p(i \land^* j)$ as the front $(k, j)$-faces. Those squares are adjacent. 
\end{proof}

We hope that the reader is starting to get a feel for proofs by composition and the interplay between two-dimensional paths and squares at this point, so we may omit uses of path abstraction without further comment, and we may also leave it up to the reader to show that the cap and tubes displayed in the composition diagrams respect the relevant adjacency conditions. 

\subsection{The groupoid laws} \label{hgroupstruct}

From the cubical perspective, path equality (homotopy) is always relative, since, viewed as squares, the only way to say that two lines are the same is modulo an identification of two other lines. By contrast, homotopy type theory~\cite{hottbook} has a globular approach. But we can simulate globular identifications of paths by considering certain squares whose remaining faces are degenerate lines. 

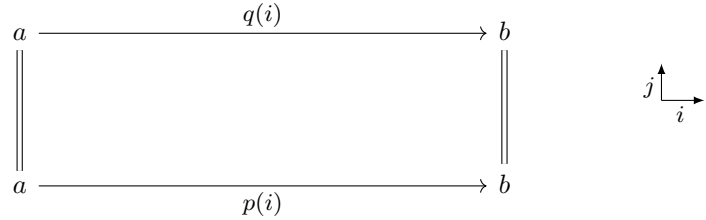
\vspace{2mm}

\begin{tikzcd}
	&&& a \arrow[rrrrrrr,"q (i)"] \arrow[ddrrrrrrr, phantom,""] &&&&&&& b && {} \\ 
	&&&&&&&&&&&& \arrow[r,shorten >= 12pt,-latex,swap,start anchor=center,"i" near start] \arrow[u,shorten >= 5pt,-latex,start anchor=center,"j" near start] & {} \\
	&&& a \arrow[rrrrrrr, swap,"p (i)"] \arrow[uu,equal] &&&&&&& b \arrow[uu,swap,equal]
\end{tikzcd}

The groupoid laws are stated using this globular representation. Put differently, when we state that reflexivity is a unit for path inversion and concatenation, that inversion is involutive, and concatenation is associative, for example, we mean that there exists a globular identification between them. 

Because the proof is simpler, we start by showing that reflexivity is a right and left unit for path concatenation in the next two lemmas.

\begin{lemma}\label{lem:ru}
For every $A$ and every $a, b : A$ we have a path
$$\mathsf{ru}_p : \path{\path{A}{a}{b}}{p}{p \sq \refl{b}}$$
for any $p : \path{A}{a}{b}$.
\end{lemma}
\begin{proof} We need to construct an $(i, j)$-square having $p (i)$ and $(p \sq \refl{b}) (i)$ as $i$-lines and $a$ and $b$ as degenerate $j$-lines. 
But this already follows from the filler of concatenation defined in the proof of Lemma~\ref{lem:comp}. 
\end{proof}

\begin{lemma}\label{lem:lu}
	For every $A$ and every $a, b : A$ we have a path
	$$\mathsf{lu}_p : \path{\path{A}{a}{b}}{p}{\trans{\refl{a}}{p}}$$
	for any $p : \path{A}{a}{b}$.
\end{lemma}
\begin{proof}
	By composition, we define a helper $(i, j)$-square that goes from $\symm{p} (i)$ to $b$ in the $i$-direction and from $b$ to $p (j)$ in the $j$-direction. The composition uses the filler of the path inversion of $p$ (front), meet (right), and degenerate squares formed from lines (back and left) and points (bottom).
	For future reference, we shall call it $\gamma$:
	
	
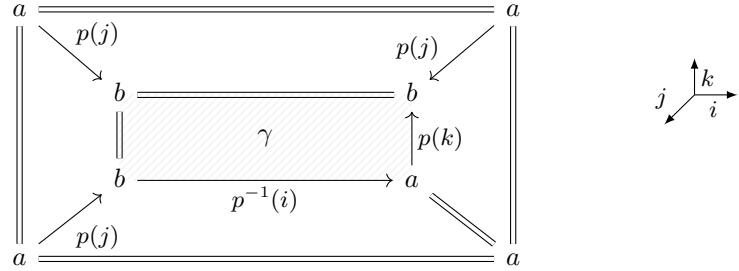
\vspace{2mm}
	
	\begin{tikzcd}
	[execute at end picture={
		\foreach \Valor/\Nombre in
		{
			tikz@f@13-2-4/a,tikz@f@13-2-8/b,tikz@f@13-3-4/c,tikz@f@13-3-8/d%
		}
		{
			\coordinate (\Nombre) at (\Valor);
		}
		\fill[pattern=north east lines,pattern color=grey,opacity=0.3]
		(b) -- (a) -- (c) -- (d) -- cycle;
	}]
		&& a \arrow[dr,"p (j)"] \arrow[rrrrrr,equal] &&&&&& a \arrow[dl,swap,"p (j)" near end] && {} \\
		&&& b \arrow[rrrr,equal] \arrow[rrrrd,phantom,"\gamma"] &&&& b &&& \arrow[r,shorten >= 12pt,-latex,swap,start anchor=center,"i" near start] \arrow[u,shorten >= 12pt,-latex,swap,start anchor=center,"k" near start] \arrow[dl,shorten >= 28pt,-latex,swap,start anchor=center,"j" near start] & {}    \\
		&&& b \arrow[rrrr,swap,"\symm{p} (i)"] \arrow[u,equal] &&&& a \arrow[u,swap,"p (k)"] && {} \\
		&& a \arrow[ur,swap,"p (j)"] \arrow[uuu,equal] \arrow[rrrrrr,swap,equal] &&&&&& a
		\arrow[ul,equal,near end] \arrow[uuu,equal]  \\
	\end{tikzcd}
	
	Forming a new open cube, we set $\gamma$ at the right, the filler of the concatenation of $\refl{a}$ and $p$ at the back, the filler of the inversion of $p$ at the bottom, and degenerate squares at the other faces.
	
	
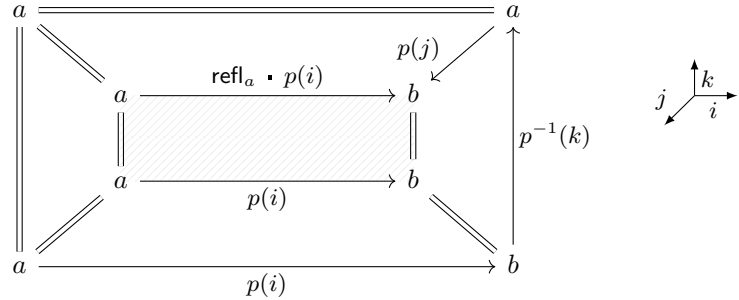
\vspace{2mm}
	
	\begin{tikzcd}
	[execute at end picture={
		\foreach \Valor/\Nombre in
		{
			tikz@f@14-2-4/a,tikz@f@14-2-8/b,tikz@f@14-3-4/c,tikz@f@14-3-8/d%
		}
		{
			\coordinate (\Nombre) at (\Valor);
		}
		\fill[pattern=north east lines,pattern color=grey,opacity=0.3]
		(b) -- (a) -- (c) -- (d) -- cycle;
	}]
		&& a \arrow[dr,equal] \arrow[rrrrrr,equal] &&&&&& a \arrow[dl,swap,"p (j)" near end] && {} \\
		&&& a \arrow[rrrr,"\trans{\refl{a}}{p} (i)"] &&&& b &&& \arrow[r,shorten >= 12pt,-latex,swap,start anchor=center,"i" near start] \arrow[u,shorten >= 12pt,-latex,swap,start anchor=center,"k" near start] \arrow[dl,shorten >= 28pt,-latex,swap,start anchor=center,"j" near start] & {}    \\
		&&& a \arrow[rrrr,swap,"p (i)"] \arrow[u,equal] &&&& b \arrow[u,swap,equal] && {} \\
		&& a \arrow[ur,swap,equal] \arrow[uuu,equal] \arrow[rrrrrr,swap,"p (i)"] &&&&&& b
		\arrow[ul,equal,near end] \arrow[uuu,swap,"\symm{p} (k)"]  \\
	\end{tikzcd}
\end{proof}

Why is the unit property so much simpler to demonstrate in the right? 
If we look attentively at the filler of path concatenation (Lemma~\ref{lem:comp}), for example, 

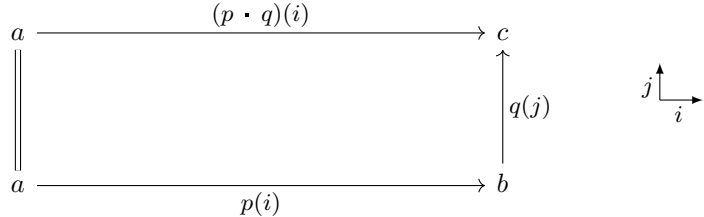
\vspace{2mm}

\begin{tikzcd}
	&&& a \arrow[rrrrrrr,"(\trans{p}{q})(i)"] \arrow[ddrrrrrrr, phantom,""] &&&&&&& c && {} \\ 
	&&&&&&&&&&&& \arrow[r,shorten >= 12pt,-latex,swap,start anchor=center,"i" near start] \arrow[u,shorten >= 5pt,-latex,start anchor=center,"j" near start] & {} \\
	&&& a \arrow[rrrrrrr, swap,"p(i)"] \arrow[uu,equal] &&&&&&& b \arrow[uu,swap,"q (j)"]
\end{tikzcd}

\vspace{2mm}

\noindent we can see that it forms a simultaneous identification that can be pronounced ``let there be a path from $p$ to $p \sq q$ just in case $q$ is $\refl{a}$". Consequently, if we set $q :\equiv \refl{b}$, we immediately have a globular path from $p$ to $p \sq q$. We can thus compare path concatenation with the transitivity operation defined in the homotopy type theory book by path induction on the second argument~\cite[lem.2.1.2]{hottbook}. The same idea applies to path inversion, ``let there be a path from $\symm{p}$ to $\refl{a}$ just in case $p$ is $\refl{a}$", so it is related to the symmetry operation defined by path induction on $p$~\cite[lem.2.1.1]{hottbook}. 

Next we prove that path inversion indeed is a right and left inverse with respect to concatenation:

\begin{lemma}\label{lem:rc}
For every $A$ and every $a, b : A$ we have a path
$$\mathsf{rc}_p :  \path{\path{A}{a}{a}}{\refl{a}}{\trans{p}{\symm{p}}}$$
for any $p : \path{A}{a}{b}$.
\end{lemma}
\begin{proof}
	By composition, we must construct a cube whose composite is an $(i, k)$-square with $\refl{a} (i)$ and $\trans{p}{\symm{p}} (i)$ as $i$-lines and $a$ in both degenerate $k$-lines. Now consider the following open $(i, j, k)$-cube 
	
	
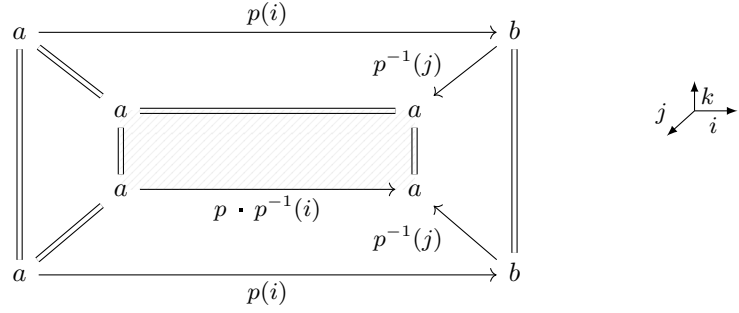
\vspace{2mm}
	
	\begin{tikzcd}
	[execute at end picture={
		\foreach \Valor/\Nombre in
		{
			tikz@f@16-2-4/a,tikz@f@16-2-8/b,tikz@f@16-3-4/c,tikz@f@16-3-8/d%
		}
		{
			\coordinate (\Nombre) at (\Valor);
		}
		\fill[pattern=north east lines,pattern color=grey,opacity=0.3]
		(b) -- (a) -- (c) -- (d) -- cycle;
	}]
		&& a \arrow[dr,equal] \arrow[rrrrrr,"p (i)"] &&&&&& b \arrow[dl,swap,"\symm{p} (j)" near end] && {} \\
		&&& a \arrow[rrrr,equal] &&&& a &&& \arrow[r,shorten >= 12pt,-latex,swap,start anchor=center,"i" near start] \arrow[u,shorten >= 12pt,-latex,swap,start anchor=center,"k" near start] \arrow[dl,shorten >= 28pt,-latex,swap,start anchor=center,"j" near start] & {}    \\
		&&& a \arrow[rrrr,swap,"\trans{p}{\symm{p}} (i)"] \arrow[u,equal] &&&& a \arrow[u,swap,equal] && {} \\
		&& a \arrow[ur,swap,equal] \arrow[uuu,equal] \arrow[rrrrrr,swap,"p (i)"] &&&&&& b
		\arrow[ul,"\symm{p} (j)" near end] \arrow[uuu,swap,equal]  \\
	\end{tikzcd}

\noindent whose bottom, left and right faces are degenerate squares, and back and front squares are respectively the fillers for path inversion and concatenation.
\end{proof}


\begin{lemma}\label{lem:lc}
	For every $A$ and every $a, b : A$ we have a path
	$$\mathsf{lc}_p : \path{\path{A}{b}{b}}{\refl{b}}{\trans{\symm{p}}{p}}$$
	for any $p : \path{A}{a}{b}$.
\end{lemma}
\begin{proof} 
	By composition on the following open cube, whose back face is the $\gamma$ square defined in the proof of Lemma~\ref{lem:lu}.
	
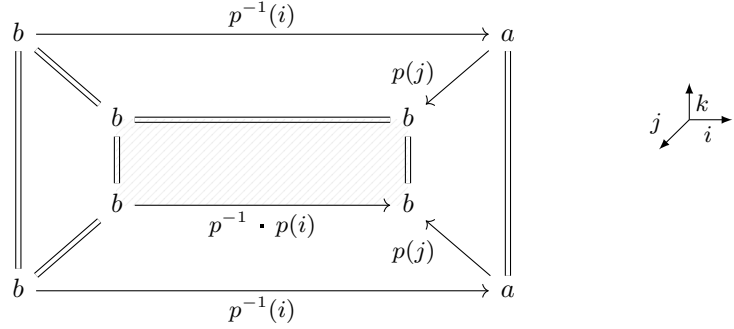
	\begin{tikzcd}
	[execute at end picture={
		\foreach \Valor/\Nombre in
		{
			tikz@f@17-2-4/a,tikz@f@17-2-8/b,tikz@f@17-3-4/c,tikz@f@17-3-8/d%
		}
		{
			\coordinate (\Nombre) at (\Valor);
		}
		\fill[pattern=north east lines,pattern color=grey,opacity=0.3]
		(b) -- (a) -- (c) -- (d) -- cycle;
	}]
		&& b \arrow[dr,equal] \arrow[rrrrrr,"\symm{p} (i)"] &&&&&& a \arrow[dl,swap,"p (j)" near end] && {} \\
		&&& b \arrow[rrrr,equal] &&&& b &&& \arrow[r,shorten >= 12pt,-latex,swap,start anchor=center,"i" near start] \arrow[u,shorten >= 12pt,-latex,swap,start anchor=center,"k" near start] \arrow[dl,shorten >= 28pt,-latex,swap,start anchor=center,"j" near start] & {}    \\
		&&& b \arrow[rrrr,swap,"\trans{\symm{p}}{p} (i)"] \arrow[u,equal] &&&& b \arrow[u,swap,equal] && {} \\
		&& b \arrow[ur,swap,equal] \arrow[uuu,equal] \arrow[rrrrrr,swap,"\symm{p} (i)"] &&&&&& a
		\arrow[ul,"p (j)" near end] \arrow[uuu,swap,equal]  \\
	\end{tikzcd}
\end{proof}


The following lemma states that path inversion is involutive:


\begin{lemma}\label{lem:invol}
For every $A$ and every $a, b : A$, we have a path
$$\mathsf{inv}_p : \path{\path{A}{a}{b}}{p}{\symm{(\symm{p})}}$$
for any $p : \path{A}{a}{b}$.
\end{lemma}
\begin{proof}
The proof follows by the use of meets, joins, and $\gamma$ to form the composite:

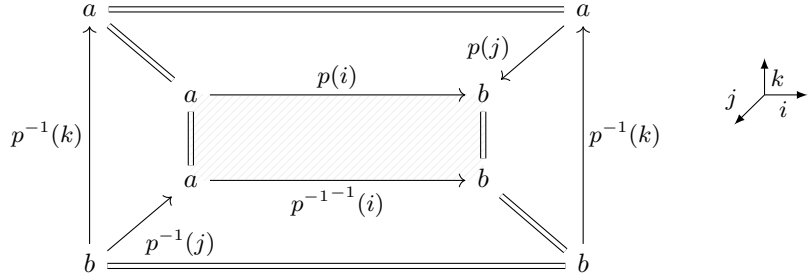
\vspace{2mm}

\begin{tikzcd}
	[execute at end picture={
		\foreach \Valor/\Nombre in
		{
			tikz@f@18-2-4/a,tikz@f@18-2-8/b,tikz@f@18-3-4/c,tikz@f@18-3-8/d%
		}
		{
			\coordinate (\Nombre) at (\Valor);
		}
		\fill[pattern=north east lines,pattern color=grey,opacity=0.3]
		(b) -- (a) -- (c) -- (d) -- cycle;
	}]
	&& a \arrow[dr,equal] \arrow[rrrrrr,equal] &&&&&& a \arrow[dl,swap,"p (j)" near end] && {} \\
	&&& a \arrow[rrrr,"p (i)"] &&&& b &&& \arrow[r,shorten >= 12pt,-latex,swap,start anchor=center,"i" near start] \arrow[u,shorten >= 12pt,-latex,swap,start anchor=center,"k" near start] \arrow[dl,shorten >= 28pt,-latex,swap,start anchor=center,"j" near start] & {}    \\
	&&& a \arrow[rrrr,swap,"\symm{\symm{p}} (i)"] \arrow[u,equal] &&&& b \arrow[u,swap,equal] && {} \\
	&& b \arrow[ur,swap,"\symm{p} (j)"] \arrow[uuu,"\symm{p} (k)"] \arrow[rrrrrr,swap,equal] &&&&&& b
	\arrow[ul,equal,near end] \arrow[uuu,swap,"\symm{p} (k)"]  \\
\end{tikzcd}
\end{proof}


Last but not least, we want to show that path concatenation is associative. 


\begin{lemma}\label{lem:assoc}
	For every $A$ and every $a, b, c, d : A$, we have a path
	$$\mathsf{assoc}_{p,q,r} : \path{\path{A}{a}{d}}{(p \sq q) \sq r}{p \sq (q \sq r)} $$
	for any $p : \path{A}{a}{b}$, $q : \path{A}{b}{c}$, $r : \path{A}{c}{d}$.
\end{lemma}
\begin{proof}
	We use the fact that, for any two squares with the same three faces, there is a square showing that the fourth sides are equal. 
	In particular, given the following two squares with definitionally equal bottom, right, and left faces
	
	
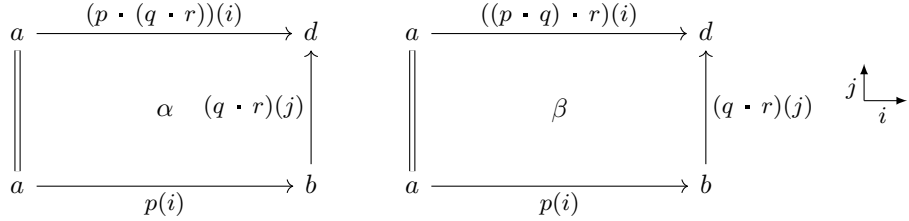
\vspace{2mm}
	\begin{tikzcd}
		a  \arrow[rrrr,"(\trans{p}{(\trans{q}{r})})(i)"] \arrow[ddrrrr, phantom,"\alpha"] &&&& d & a \arrow[rrrr,"(\trans{(\trans{p}{q})}{r})(i)"] \arrow[ddrrrr, phantom,"\beta"] &&&& d && {} \\
		&&&&& &&&&  && \arrow[r,shorten >= 12pt,-latex,swap,start anchor=center,"i" near start] \arrow[u,shorten >= 5pt,-latex,start anchor=center,"j" near start] & {} \\
		a \arrow[uu,equal] \arrow[rrrr,swap,"p (i)"] &&&& b \arrow[uu,"(\trans{q}{r})(j)"] & a \arrow[rrrr,swap,"p (i)"] \arrow[uu,equal] &&&& b \arrow[uu,swap,"(\trans{q}{r})(j)"]
	\end{tikzcd}
	
	\vspace{2mm}
	
	\noindent 
	we can construct the desired path homotopy
	
	
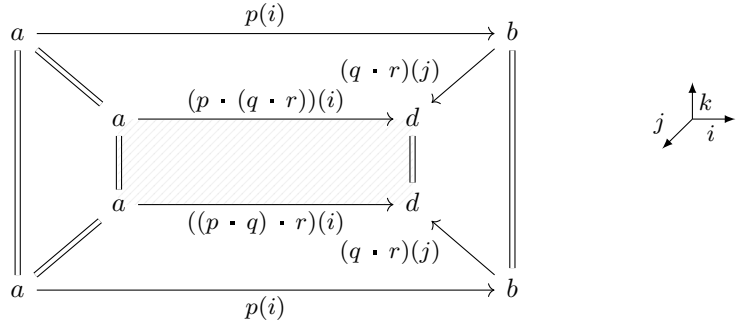
\begin{tikzcd}
	[execute at end picture={
		\foreach \Valor/\Nombre in
		{
			tikz@f@20-2-4/a,tikz@f@20-2-8/b,tikz@f@20-3-4/c,tikz@f@20-3-8/d%
		}
		{
			\coordinate (\Nombre) at (\Valor);
		}
		\fill[pattern=north east lines,pattern color=grey,opacity=0.3]
		(b) -- (a) -- (c) -- (d) -- cycle;
	}]
		&& a \arrow[dr,equal] \arrow[rrrrrr,"p (i)"] &&&&&& b \arrow[dl,swap,"(\trans{q}{r})(j)" near end] && {} \\
		&&& a \arrow[rrrr,"(\trans{p}{(\trans{q}{r})})(i)"] &&&& d &&& \arrow[r,shorten >= 12pt,-latex,swap,start anchor=center,"i" near start] \arrow[u,shorten >= 12pt,-latex,swap,start anchor=center,"k" near start] \arrow[dl,shorten >= 28pt,-latex,swap,start anchor=center,"j" near start] & {}    \\
		&&& a \arrow[rrrr,swap,"(\trans{(\trans{p}{q})}{r})(i)"] \arrow[u,equal] &&&& d \arrow[u,equal] && {} \\
		&& a \arrow[ur,swap,equal] \arrow[uuu,equal] \arrow[rrrrrr,swap,"p (i)"] &&&&&& b
		\arrow[ul,"(\trans{q}{r})(j)" near end] \arrow[uuu,swap,equal]  \\
	\end{tikzcd}
	
	\vspace{2mm}
	
	\noindent
	Since $\alpha$ is just the filler of path concatenation $\trans{p}{(\trans{q}{r})}$, it remains to construct $\beta$. But this is easy, because looking at the top and right sides of this square, the filler of path concatenation gives us a canonical construction:
	
	
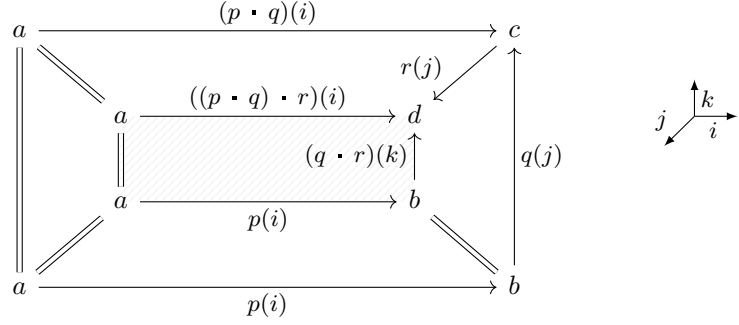
\vspace{2mm}
	
	\begin{tikzcd}
	[execute at end picture={
		\foreach \Valor/\Nombre in
		{
			tikz@f@21-2-4/a,tikz@f@21-2-8/b,tikz@f@21-3-4/c,tikz@f@21-3-8/d%
		}
		{
			\coordinate (\Nombre) at (\Valor);
		}
		\fill[pattern=north east lines,pattern color=grey,opacity=0.3]
		(b) -- (a) -- (c) -- (d) -- cycle;
	}]
		&& a \arrow[dr,equal] \arrow[rrrrrr,"(\trans{p}{q})(i)"] &&&&&& c \arrow[dl,swap,"r (j)" near end] && {} \\
		&&& a \arrow[rrrr,"(\trans{(\trans{p}{q})}{r})(i)"] &&&& d &&& \arrow[r,shorten >= 12pt,-latex,swap,start anchor=center,"i" near start] \arrow[u,shorten >= 12pt,-latex,swap,start anchor=center,"k" near start] \arrow[dl,shorten >= 28pt,-latex,swap,start anchor=center,"j" near start] & {}    \\
		&&& a \arrow[rrrr,swap,"p (i)"] \arrow[u,equal] &&&& b \arrow[u,"(\trans{q}{r})(k)"] && {} \\
		&& a \arrow[ur,swap,equal] \arrow[uuu,equal] \arrow[rrrrrr,swap,"p (i)"] &&&&&& b
		\arrow[ul,equal, near end] \arrow[uuu,swap,"q (j)"]  \\
	\end{tikzcd}
\end{proof}

What about the higher groupoid structure of types? What we have shown in this section is that types have a 1-groupoid structure, but since those laws do not hold ``on the nose" as definitional equations, they must satisfy some equations of their own. 
We will not cover this here. Instead we give a proof of path induction in \S\ref{pathinduction}, which corresponds almost exactly to the elimination rule for the inductive path type in homotopy type theory~\cite{hottbook} (except that the computation rule here only holds up to a path), and conjecture that it should be possible to derive them with an approach similar to~\cite{lumsdaine2010weak}.


\section{Dependent paths} \label{sec:dependent}

The various operations and laws defined in the previous section all share a fundamental limitation: they are only applicable to a specific class of paths. This restriction is imposed by our implicit requirement that a path be non-dependent, which means that the type $\path{A}{a}{b}$ is not well-formed unless $a$ and $b$ have exactly the type $A$. While this restriction is important when trying to understand cubical methods, it rules out the formation of paths in paths between types $A : \I \to \U$, which, for the lack of a better name, we shall call \textit{type lines}, types that may change depending on their endpoints. 

\subsection{The dependent path type} \label{sec:pathd}

Given a type line $A : \I \to \U$ and terms $a : A(\zero)$ and $b : A(\one)$, the type of dependent paths from $a$ to $b$ is written $\pathd{A}{a}{b}$.
If the underlying type line is a constant function, that is, the reflexivity path, then the dependent path type is the non-dependent path type
$$\pathd{\fun \_.A}{a}{b} \equiv \path{A}{a}{b}.$$
The only reason why this equality holds is because non-dependent paths are actually defined in terms of dependent paths. So the rules for the dependent path type arise as straightforward generalizations of the ones stated before for path types in virtually the same way the dependent function type generalizes the non-dependent function type. 
For instance, as with non-dependent paths, we eliminate a term of this type by application, $p (i) : A(i)$, where $p : \pathd{A}{a}{b}$ and $i : \I$. The constructor of this type is path abstraction, and we say that $\pabs{i}{a} : \pathd{A}{a\subst{\zero}{i}}{a\subst{\one}{i}}$ if $a : A(i)$ assuming that $i : \I$, where $a\subst{\zero}{i} : A(\zero)$ and $a\subst{\one}{i} : A(\one)$. 
Moreover, we also require equalities that are no different from the ones for the ordinary path type. In particular, given~$p : \pathd{A}{a}{b}$ and~$i : \I$, we have boundary rules $p (\zero) \equiv a : A(\zero)$ and $p (\one) \equiv b : A(\one)$, and the uniqueness rule $\pabs{i}{(p(i))}  \equiv p$ when $i$ does not occur in $p$. 

We need a type for dependent paths in the cubical setting because in their most general form paths connect terms inhabiting the endpoints of a type line. 
Consider for example the meet operator defined in Lemma~\ref{lem:meet}. What, exactly, is its type? Starting with an iterated function type

$$\fun j. \fun i. p (i \land j) : \I \to \I \to A$$ 
we can obtain a path in the type of paths from $a$ to $p (j)$
$$\fun j. \pabs{i}{p (i \land j)} : \dfun{j}{\I}{\path{A}{a}{p (j)}},$$
because, given a fixed $j$, $\fun i.p(i \land j)$ varies from the point $a$ to $p (j)$. By repeating the process, we obtain a dependent path from $\refl{a}$ to $p$ in the type of paths from $a$ to $p (j)$
$$\pabs{j}{\pabs{i}{p (i \land j)}} : \pathd{\fun i.\path{A}{a}{p (i)}}{\refl{a}}{p}$$
since $\fun j.\pabs{i}{p(i \land j)}$ goes from 

\begin{align*}
\pabs{i}{p(i \land \zero)} \equiv & \pabs{i}{a}
\end{align*}
to
\begin{align*}
\pabs{i}{p(i \land \one)} \equiv & \pabs{i}{p (i)} \\
\equiv & p.
\end{align*}

In homotopy type theory~\cite{hottbook} a type for dependent paths is definable using the non-dependent and inductively defined path type. The advantage of the cubical approach to dependent paths is that they are naturally built-in. It is also worth noting that a type for dependent paths can be defined as multi-dimensional extension types instead~\cite[p.67]{angiuli2019computational}, a generalization of path types that avoid the intricacies that are innate to the use of iterative path types. 

\subsection{Heterogeneous composition} \label{sec:groupoidpathd}

As soon as we start considering dependent paths it becomes clear that the path inversion and concatenation operations defined in \S\ref{hop} are not general~enough. 
Let us consider the limitations of non-dependent path inversion first. 
Assuming that $\pathd{A}{a}{b}$ is a type, where $a : A(\zero)$ and $b : A(\one)$, the type $\pathd{A}{b}{a}$ of inverse paths from $b$ to $a$ will \textit{not} be well-formed unless it is also the case that $b : A(\zero)$ and $a : A(\one)$. For non-dependent paths this condition turns out to be trivially satisfied since the relevant type line is a constant function, meaning that $A(\zero) \equiv A \equiv A(\one)$. But as this is not true in general, the operation defined in Lemma~\ref{lem:inv} fails to produce inverses for dependent paths. 

Fortunately, there is a way to deal with this problem using type universes. Suppose we are given a type line $A : \I \to \U$, that is, a non-dependent path from the type $A(\zero)$ to $A(\one)$ in the universe $\U$,
$$A : \path{\U}{A(\zero)}{A(\one)}.$$
We can assume that this is a non-dependent path because type universes never depend on interval variables, so we have both $A(\zero), A(\one) : \U$.

Now it can be shown by Lemma~\ref{lem:inv} that the following inverse exists:

\vspace{2mm}

\begin{tikzcd}
	&& A(\one) \arrow[rrrrrr,dotted,"\symm{A}(i)"] \arrow[ddrrrrrr, phantom,""] &&&&&& A(\zero) && {} \\ 
	&&&&&&&&&& \arrow[r,shorten >= 12pt,-latex,swap,start anchor=center,"i" near start] \arrow[u,shorten >= 5pt,-latex,start anchor=center,"j" near start] & {} \\
	&& A(\zero) \arrow[rrrrrr, swap,equal] \arrow[uu,"A(j)"] &&&&&& A (\zero) \arrow[uu,swap,equal]
\end{tikzcd}

\vspace{2mm}

\noindent In particular, we have
$$\symm{A} : \path{\U}{A(\one)}{A(\zero)},$$
which corresponds precisely to the ``inverse" type of $A$, for the initial and terminal points of $\symm{A}$ are respectively $A(\one)$ and $A(\zero)$. Put differently, we have two inferences that hold top/bottom and bottom/top,
$$\Efrac{a : \symm{A} (1)}{a : A(\zero)} \quad\text{and}\quad \Efrac{a : \symm{A} (0)}{a : A(\one)}.$$
Thus, assuming that $\pathd{A}{a}{b}$ is a type, it is now easy to see that $\pathd{\symm{A}}{b}{a}$ will always be a type as well, regardless of whether $A$ is a degenerate line or not: if $\pathd{A}{a}{b}$ is a type then we have $a : A(\zero)$ and $b : A(\one)$, meaning that $a : \symm{A} (1)$ and $b : \symm{A} (0)$ must be the case.

This motivates the definition of a dependent path inversion operation:

\begin{lemma}\label{lem:hinv}
	For every type $A : \I \to \U$ and every $a : A(\zero)$ and $b : A(\one)$, there is a function
	$$\pathd{A}{a}{b} \to \pathd{{\symm{A}}}{b}{a}$$
	called the (dependent) inverse path and denoted $p \mapsto \symm{p}$.
\end{lemma}

\begin{proof}
	Dependent operations are typically introduced by repeating, mutatis mutandis, the proofs of their non-dependent counterparts. So, the results which were previously proven by composition will now follow from a ``heterogeneous" variant of composition based on the same open cubes. This adaptation is required because we are dealing with non-degenerate type lines and dependent paths, and, for a valid composition, all the faces of an open cube, that is, all the terms of the composition scenario, must have the same type. 
	
	More concretely, the open box used in the definition of non-dependent path inversion (the one used in the proof of Lemma~\ref{lem:inv}) is ill-formed in this context since $p (j) : A(j)$ and $a : A(\zero)$ are lines with different types:
	
	\vspace{2mm}
	\begin{tikzcd}
		&& b \arrow[ddrrrrrrr, phantom,""] &&&&&&& a && {} \\ 
		&&&&&&&&&&& \arrow[r,shorten >= 12pt,-latex,swap,start anchor=center,"i" near start] \arrow[u,shorten >= 5pt,-latex,start anchor=center,"j" near start] & {} \\
		&& a \arrow[rrrrrrr, swap,equal] \arrow[uu,"p (j)"] &&&&&&& a \arrow[uu,swap,equal]
	\end{tikzcd}
	\vspace{2mm}
	
	\noindent The solution is not to form an open box living in~$A$, but rather to perform a composition taking place in~$\symm{A}$, using transport on the type line $\fun j.\filler_{j}(\symm{A}(i))$, the filler of the inverse of $A$, to adjust the type of the faces of open box.

	For the cap, we want to transport $a$ from $\zero$ to $\one$. This gives us the correct type because $\filler_{\one}(\symm{A}(i))\equiv \symm{A}(i)$. Now we observe that $a : \filler_{\zero}(\symm{A}(i))\equiv {A}(0)$, so we have
	
	$${a}^{\lto{\zero}{\one}}_{\fun j.\filler_{j}(\symm{A}(i))} : \symm{A}(i).$$
	Now for the left tube, i.e. assuming that $i=\zero$, we transport $p$ from $j$ to $\one$. More specifically, since $p (j) : \filler_{j}(\symm{A}(\zero))\equiv {A}(j)$, and, again, $\filler_{\one}(\symm{A}(\zero))\equiv \symm{A}(\zero)$, we have
	
	$${(p (j))}^{\lto{j}{\one}}_{\fun j.\filler_{ j}(\symm{A}(\zero))} : \symm{A}(\zero).$$
	The right tube is treated similarly, except that we transport $a$ (and $i=\one$),
	
	$${a}^{\lto{j}{\one}}_{\fun j.\filler_{ j}(\symm{A}(\one))} : \symm{A}(\one).$$
	
	It is not hard to show that those lines are adjacent. We want to form a composition from $\zero$ to $\one$, so the initial point of the left tube should match the initial point of the cap, 
	
	\begin{align*}
	(\fun j. {(p (j))}^{\lto{j}{\one}}_{\fun j.\filler_{ j}(\symm{A}(\zero))})(\zero) 
		&\equiv {(p (\zero))}^{\lto{\zero}{\one}}_{\fun j.\filler_{j}(\symm{A}(\zero))} \\
		&\equiv {a}^{\lto{\zero}{\one}}_{\fun j.\filler_{j}(\symm{A}(\zero))} \\
		&\equiv {a}^{\lto{\zero}{\one}}_{\fun j.\filler_{j}(\symm{A}(i))}
	\end{align*}
	
	\noindent and the initial point of the right tube should be the terminal point of the cap,
	
	\begin{align*}
	(\fun j. {a}^{\lto{j}{\one}}_{\fun j.\filler_{ j}(\symm{A}(\one))})(\zero) 
		&\equiv  {a}^{\lto{\zero}{\one}}_{\fun j.\filler_{j}(\symm{A}(\one))} \\
		&\equiv {a}^{\lto{\zero}{\one}}_{\fun j.\filler_{j}(\symm{A}(i))}.
	\end{align*}
	
	By composition, we have a line from $b$ to $a$ in $\symm{A}(i)$. This composite has the correct endpoints because 
	
	$${(p (\one))}^{\lto{\one}{\one}}_{\fun j.\filler_{ j}(\symm{A}(\zero))}\equiv p(\one) \equiv b$$
	and
	$${a}^{\lto{\one}{\one}}_{\fun j.\filler_{j}(\symm{A}(\one))} \equiv a.$$
\end{proof}

It is important to stress that the proof of the previous lemma is based on \textit{heterogeneous composition}, a particular kind of composition in which the types of the cap and composite may differ. As indicated in the proof above, a heterogeneous composition can be obtained from composition and transport. For instance, given a type line $A : \I \to \U$, a heterogeneous composition in $A$ with cap $a : \I \to A(\zero)$ and tubes $a_0, a_1 : \prod_{(i : \I)} A(i)$ is just an abbreviation for a compound composition which combines the two Kan operations into one, a composition with cap $\fun i.{(a(i))}^{\lto{\zero}{\one}}_{A} : \I \to A(\one)$ and tubes $\fun j.{(a_0(j))}^{\lto{j}{\one}}_A : \I \to A(\one)$ and $\fun j.{(a_1(j))}^{\lto{j}{\one}}_A : \I \to A(\one)$. The composite is a term of type $\I \to A(\one)$. 


So we have already covered the case of path inversion and showed how we can invert dependent paths using heterogeneous composition. One natural question is whether a similar strategy can be used to concatenate dependent paths. It is easier to see the limitations of the non-dependent operation of path concatenation when we consider two-dimensional paths such as 

$$\alpha : \pathd{\fun j.\path{A}{s (j)}{t (j)}}{p}{q} \qquad\text{and}\qquad \beta : \pathd{\fun j.\path{A}{u (j)}{v (j)}}{q}{r},$$
which correspond to two $(i,j)$-squares in $A$ such as the following:

\vspace{2mm}
\begin{tikzcd}
	c \arrow[rrrr,"q(i)"] \arrow[ddrrrr, phantom, "\alpha(j)(i)"] &&&& d & e \arrow[rrrr,"r(i)"] \arrow[ddrrrr, phantom, "\beta(j)(i)"] &&&& f & {} \\ 
	&&&&&&&&&& \arrow[r,shorten >= 12pt,-latex,swap,start anchor=center,"i" near start] \arrow[u,shorten >= 5pt,-latex,start anchor=center,"j" near start] & {} \\
	a \arrow[rrrr, swap,"p (i)"] \arrow[uu,swap,"s (j)"] &&&& b \arrow[uu,"t (j)"] & c \arrow[uu,swap,"u (j)"] \arrow[rrrr,swap,"q(i)"] &&&& d \arrow[uu,"v (j)"] 
\end{tikzcd}
\vspace{2mm}

\noindent It is clear that we should be able to concatenate the paths $\alpha$ and $\beta$ vertically, using the bottom and top faces of the squares, if we concatenate their left and right faces using non-dependent path concatenation. This should give us a $j$-dependent path from $p$ to $r$ in the type of paths from $s \sq u(j)$ to $t \sq v(j)$ in $A$. This construction is illustrated in the following composition, whose left and right squares can be easily obtained by composition using the fillers of path concatenation and inversion, and meets and joins:  

\begin{tikzcd}
	[execute at end picture={
		\foreach \Valor/\Nombre in
		{
			tikz@f@21-2-4/a,tikz@f@21-2-8/b,tikz@f@21-3-4/c,tikz@f@21-3-8/d%
		}
		{
			\coordinate (\Nombre) at (\Valor);
		}
		\fill[pattern=north east lines,pattern color=grey,opacity=0.3]
		(b) -- (a) -- (c) -- (d) -- cycle;
	}]
	&& c \arrow[dr,"u(j)"] \arrow[rrrrrr,"q (i)"] &&&&&& d \arrow[dl,swap,"r (j)" near end] && {} \\
	&&& e \arrow[rrrr,"r(i)"] \arrow[rrrrd,phantom,"\alpha ? \beta (j) (i)"] &&&& f &&& \arrow[r,shorten >= 12pt,-latex,swap,start anchor=center,"i" near start] \arrow[u,shorten >= 12pt,-latex,swap,start anchor=center,"k" near start] \arrow[dl,shorten >= 28pt,-latex,swap,start anchor=center,"j" near start] & {}    \\
	&&& a \arrow[rrrr,swap,"p(i)"] \arrow[u,"{s \sq u}(k)"] &&&& b \arrow[u,swap,"{t \sq v}(k)"] && {} \\
	&& c \arrow[ur,swap,"{s}^{-1}" near end] \arrow[uuu,equal] \arrow[rrrrrr,swap,"p (i)"] &&&&&& d
	\arrow[ul,"{t}^{-1} (j)" near end] \arrow[uuu,swap,equal]  \\
\end{tikzcd}

\noindent This two-dimensional concatenation is not an instance of the non-dependent operation derived in Lemma~\ref{lem:comp}: from our definition, it follows that to form the concatenated path $\trans{\alpha}{\beta}$, first $\alpha$ and $\beta$ must be terms of the same type and that type must be a degenerate line. 
In this example, the concatenation fails for both reasons, for the equality $\path{A}{s (j)}{t (j)} \equiv \path{A}{u (j)}{v (j)}$ need \textit{not} be the case and those types may depend on the interval variable $j$. 

To overcome this problem we need a dependent path concatenation operation. Following the definition of dependent path inversion, we consider the non-dependent path concatenation of line types first. Suppose that $A,B : \I \to \U$ such that $A(\one) \equiv B (0)$.
The diagram 

\vspace{2mm}

\begin{tikzcd}
	& A(\zero) \arrow[rrrrrr,dotted,"(A \sq B) (i)"] \arrow[ddrrrrrr, phantom,""] &&&&&& B (1) && {} \\ 
	&&&&&&&&& \arrow[r,shorten >= 12pt,-latex,swap,start anchor=center,"i" near start] \arrow[u,shorten >= 5pt,-latex,start anchor=center,"j" near start] & {} \\
	& A(\zero) \arrow[rrrrrr, swap,"A(i)"] \arrow[uu,equal] &&&&&& A(\one) \equiv B (0) \arrow[uu,swap,"B(j)"]
\end{tikzcd}

\vspace{2mm}

\noindent illustrates the path concatenation of $A$ and $B$, 
$$ A \sq B \; : \; \path{\U}{A(\zero)}{B(\one)}.$$
We also have two important inferences that hold top/bottom and bottom/top,
$$\Efrac{a : (A \sq B) (0)}{a : A(\zero)} \quad\text{and}\quad \Efrac{b : (A \sq B) (1)}{b : B (1)}.$$

With this, we have all we need to define concatenation for dependent paths:


\begin{lemma} \label{lem:hcomp}
	Suppose that $A, B : \I \to \U$ such that $A(\one) \equiv B (0)$. Given any $a : A(\zero)$, $b : A(\one)$ and $c : B (1)$, there is a function
	$$\pathd{A}{a}{b} \to \pathd{B}{b}{c} \to \pathd{\trans{A}{B}}{a}{c}$$
	written $p \mapsto q \mapsto p \,\sq q$ and called the dependent path concatenation function.
\end{lemma}
\begin{proof}
	By heterogeneous composition on the open box from Lemma~\ref{lem:comp}.
\end{proof}


Dependent path concatenation allows us to concatenate $\alpha$ and $\beta$ from our example above, but it is worth noting that the resulting path need not be definitionally equal to the operation $\alpha ? \beta$ we described. However, it is easy to show by path induction that they are equal up to globular identification.

\subsection{The groupoid laws for dependent paths} \label{heterogroupstruct}

Now that we are in possession of a new pair of dependent path inverse and concatenation operations, it is possible for us to revisit the groupoid structure from~\S\ref{hgroupstruct} and derive more general laws that hold for dependent paths as well. It is helpful to understand how this generalization works with an example. 

As a case in point, let us examine the involution law from Lemma~\ref{lem:invol}. This law actually expresses a fact about paths in degenerate line types. To be exact, given a degenerate $A : \I \to \U$, it states that, for every $a : A(\zero)$, $b : A(\one)$ and $p : \pathd{A}{a}{b}$, there is a path 

$$\pathd{\fun \_. \pathd{A}{a}{b}}{\symm{(\symm{p})}}{p}.$$ 

\noindent If we were to drop the restriction that the type line $A$ be degenerate, meaning that the law would then be founded on dependent path inversion instead, the desired type would have to be stated as something like 

$$\pathd{\fun j.\pathd{\fun i. ?(j)(i)}{a}{b}}{\symm{(\symm{p})}}{p}.$$

\noindent Because we are dealing with dependent paths and their operations, which often result in paths living in a different type line, we run into a problem at this point: we need a path from $\symm{(\symm{p})} : \pathd{\symm{(\symm{A})}}{a}{b}$ to $p : \pathd{A}{a}{b}$, and we must specify in what type line this path lives in. Now, considering that generally $\symm{(\symm{A})} \not\equiv A$, this type line cannot be degenerate. 
But recall that, by the non-dependent involution law from Lemma~\ref{lem:invol}, there is a path 

$$\mathsf{inv}_A : \path{\path{\U}{A(\zero)}{A(\one)}}{\symm{(\symm{A})}}{A},$$

\noindent so the law we are looking for can be stated as follows: 

\begin{lemma}\label{lem:hinvol}
	For every $A : \I \to \U$ with $a : A(\zero)$ and $b : A(\one)$, we have
	$$\pathd{\fun j.\pathd{\fun i. \mathsf{inv}_A(j) (i)}{a}{b}}{\symm{(\symm{p})}}{p}$$
	for any dependent path $p : \pathd{A}{a}{b}$. 
\end{lemma}

The proof argument is similar to those given in the definition of dependent path inversion and composition from Lemmas~\ref{lem:hinv} and \ref{lem:hcomp}, respectively. It thus follows from a straightforward heterogeneous composition on the open cube constructed for the proof of the non-dependent version of the involution law from Lemma~\ref{lem:invol}. In fact, all dependent counterparts of the laws proven in~\S\ref{hgroupstruct} follow the same pattern: they can all be stated by using their non-dependent counterparts and proven by a heterogeneous filling of their open cubes. For this reason we will simply omit those results. 

\section{Notable properties of paths} \label{moreabout}

Before concluding our naive presentation of cubical type theory, we would like to explore two notable properties of paths from a cubical perspective: path induction and the fact that path concatenation operation on the second loop space is commutative.

\subsection{Path induction} \label{pathinduction}

We now present a derivation of path induction~\cite[\S1.12.1]{hottbook}, a key property that informally states that paths with a free endpoint can be deformed and retracted without changing their essential characteristics. As mentioned before, path induction serves as the elimination rule for the inductive path type in homotopy type theory, where it is taken as primitive. The result considered here corresponds to what is sometimes called based path induction~\cite{hottbook}: 

\begin{theorem}[Path induction]
	Given a type $A : \mathcal{U}$, a term $a : A$ and a type family $C : \prod_{(x : A)} \path{A}{a}{x} \to \U$, we have a function
	$$\mathsf{pathrec} : \prod_{(x : A)} \prod_{(p : \path{A}{a}{x})} \prod_{(c : C(a,\refl{a}))} C(x,p).$$
\end{theorem}
\begin{proof}
	We want to construct, for every $x : A $, $p : \path{A}{a}{x}$ and $c : C(a,\refl{a})$, a term of type $C(x,p)$. To obtain such a term, we shall transport $c$ along a type line $D : \mathbb{I} \to \U$ that goes from $D(\zero) :\equiv C(a,\refl{a})$ to $D(\one) :\equiv C(x,p)$. 
	For the construction of the desired type line $D$ it suffices to consider the halfway meet connection derived in our proof of Lemma~\ref{lem:meet}: 
		
	
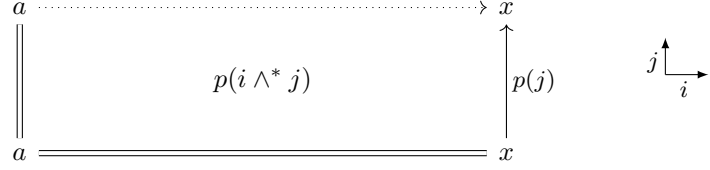
\vspace{2mm}
	
	\begin{tikzcd}
		&& a \arrow[rrrrrrr,dotted] \arrow[ddrrrrrrr,phantom,"p (i \land^* j)"] &&&&&&& x && {} \\ 
		&&&&&&&&&&& \arrow[r,shorten >= 12pt,-latex,swap,start anchor=center,"i" near start] \arrow[u,shorten >= 5pt,-latex,start anchor=center,"j" near start] & {} \\
		&& a \arrow[rrrrrrr,swap,equal] \arrow[uu,equal] &&&&&&& x \arrow[uu,swap,"p (j)"]
	\end{tikzcd}
	
	\vspace{2mm}
	
	\noindent Now it suffices to define our type line as: 
	
	$$D :\equiv \fun {i}. C(p (i \land^* \one), \pabs{j}{p (i \land^* j)}) : \mathbb{I} \to \mathcal{U}$$
	
	\noindent because it goes from
	
	\begin{align*}
	D(\zero)	&\equiv (\fun {i}. C(p (i \land^* \one), \pabs{j}{p (i \land^* j)})) \zero \\
	&\equiv C(p (\zero \land^* \one), \pabs{j}{p (\zero \land^* j)}) \\
	&\equiv C(a,\refl{a})
	\end{align*}
	
	\noindent to
	
	\begin{align*}
	D(\one) 	&\equiv (\fun {i}. C(p (i \land^* \one), \pabs{j}{p (i \land^* j)})) \one \\
	&\equiv C(p (\one \land^* \one), \pabs{j}{p (\one \land^* j)}) \\	
	&\equiv C(x, \pabs{j}{p (j)}) \\
	&\equiv C(x, p).
	\end{align*}
	
	\noindent To complete the proof we transport $c : D(\zero)$ from $\zero$ to $\one$.
	
\end{proof}

Note that, in the above proof, the $\eta$-rule for the path type discussed in \S\ref{path}, that is, the requirement that $p \equiv \pabs{j}{(p (j))}$, is crucial to the correct specification of endpoints of the type line we are doing the transportation over. Without this rule, it would not be possible to show that $D(\one) \equiv C(x, p)$ and the transportation would give us a term of the wrong type.

In cubical type theory, the computation rule for path induction does not hold ``on the nose" like in the case of the eliminator of the inductive path type~\cite{hottbook}. Put differently, for a fixed type family $C : \prod_{(x : A)} \path{A}{a}{x} \to \U$, given $a : A $ and $c : C(a,\refl{a})$, in general,
$$\mathsf{pathrec}(a,\refl{a},c) \not\equiv c.$$
But this equality can be shown to hold up to a path:

\begin{lemma}[Path computation~\cite{angiuli2019computational}] 
	For every $a : A $ and $c : C(a,\refl{a})$, we have a path of type 
	$$\path{C(a, \refl{a})}{\mathsf{pathrec}(a,\refl{a},c)}{c}.$$
\end{lemma}
\begin{proof}
	
	By the definition of path induction, we have a definitional equality

	$$\mathsf{pathrec}(a, \refl{a}, c) \equiv {c}^{\lto{\zero}{\one}}_{\fun {i}. C(\refl{a} (i \land^* \one), \pabs{j}{\refl{a} (i \land^* j)})}.$$

	\noindent We are to find a path from this transported term to $c$ in the type $C(a, \refl{a})$. First, we note that the transportation induces a path
	
	$$\pabs{i}{{c}^{\lto{i}{1}}_{\fun {i}. C(\refl{a} (i \land^* \one), \pabs{j}{\refl{a} (i \land^* j)})}}.$$
	
	\noindent This path has the right endpoints, since the static transportation from $\one$ to $\one$ has no effect, being therefore definitionally equal to $c$. But we need a path in the type $C(a, \refl{a})$ and when we consider the type line the transportation in the constructed path occurs along, it can be seen that this induced path is actually in the type $C(\refl{a} (i \land^* \one), \pabs{j}{\refl{a} (i \land^* j)})$. This is not the type $C(a, \refl{a})$ we are looking for because $a \not\equiv \refl{a} (i \land^* \one)$ and $\refl{a} \not\equiv \fun{j}. \refl{a} (i \land^* j)$.

	To conclude this proof we will fix this type mismatch with a heterogeneous composition based on the following composition scenario, where the composite is intended to give us the desired path in the type $C(a, \refl{a})$ : 
	
	\vspace{2mm}
	
	\begin{tikzcd}
		& {c}^{\lto{\zero}{\one}}_{\fun {i}. C({\refl{a}(i \land^* \one)}, \pabs{j}{\refl{a} (i \land^* j)})} \arrow[rrr,dotted] \arrow[ddrrr,phantom,""] &&& c && {} \\ 
		&&&&&& \arrow[r,shorten >= 12pt,-latex,swap,start anchor=center,"i" near start] \arrow[u,shorten >= 5pt,-latex,start anchor=center,"j" near start] & {} \\
		& c \arrow[rrr,swap,equal] \arrow[uu,"{c}^{\lto{\zero}{j}}_{\fun {i}. C({\refl{a}(i \land^* \one)}, \pabs{j}{\refl{a} (i \land^* j)})}"] &&& c \arrow[uu,swap,equal]
	\end{tikzcd}
	
	\vspace{2mm}
	
	\noindent However, before this can be viewed as a valid heterogeneous composition, we must specify in what type line the operation takes place. We therefore have to define a $D_i : \I \to \U$ containing the cap and composite as terms inhabiting its initial and terminal endpoints. More explicitly, we need to satisfy the endpoint conditions that $D_i(\zero) \equiv C(a, \refl{a}) \equiv D_i(\one)$. In the $i$-direction, however, we need $D_\zero \equiv \fun{j}. C({\refl{a}(j \land^* \one)}, \pabs{k}{\refl{a} (j \land^* k)}) $ and $D_\one \equiv \fun{j}.C(a, \refl{a})$. 
		
	First we construct by composition an otherwise degenerated square with $\refl{a}(k \land^* \one)$ as the right face. The composition scenario consists of degenerate squares in the bottom, front, back, and left and $\refl{a}(k \land^* \one)$ in the right: 
	
	\vspace{2mm}

		\begin{tikzcd}
			[execute at end picture={
				\foreach \Valor/\Nombre in
				{
					tikz@f@29-2-4/a,tikz@f@29-2-8/b,tikz@f@29-3-4/c,tikz@f@29-3-8/d%
				}
				{
					\coordinate (\Nombre) at (\Valor);
				}
				\fill[pattern=north east lines,pattern color=grey,opacity=0.3]
				(b) -- (a) -- (c) -- (d) -- cycle;
			}]
			&& a \arrow[ddd,phantom,"{\refl{a}(k \land^* j)}"{sloped,yshift=15pt} near start] \arrow[dr,equal] \arrow[rrrrrr,equal] &&&&&& a \arrow[dl,equal] && {} \\
			&&& a \arrow[rrrr,equal] \arrow[drrrr,phantom,"{\alpha(k)(i)}"] &&&& a &&& \arrow[r,shorten >= 12pt,-latex,swap,start anchor=center,"i" near start] \arrow[u,shorten >= 12pt,-latex,swap,start anchor=center,"k" near start] \arrow[dl,shorten >= 28pt,-latex,swap,start anchor=center,"j" near start] & {}    \\
			&&& a \arrow[rrrr,equal] \arrow[u] &&&& a \arrow[u,swap,equal] && {} \\
			&& a \arrow[ur,swap,equal] \arrow[uuu,equal] \arrow[rrrrrr,equal] &&&&&& a
			\arrow[ul,equal] \arrow[uuu,swap,equal]  \\ 
		\end{tikzcd}
	
	\vspace{2mm}
	
	\noindent Finally, we set $D_i :\equiv \fun {k}. C(\alpha(k)(i), \fun{j}.\filler_{j}(\alpha(k)(i)))$. Now, to check that the type line $D_i$ has the correct endpoints we observe that for either $\mathsf{\epsilon}=\zero$ or $\mathsf{\epsilon}=\one$:
	
	\begin{align*}
	D_i(\mathsf{\epsilon}) 
	&\equiv (\fun{i}. C(\alpha(k)(i), \fun{j}.\filler_{j}(\alpha(k)(i))))(\mathsf{\epsilon}) \\	
	&\equiv C(\alpha(\mathsf{\epsilon})(i), \fun{j}.\filler_{j}(\alpha(\mathsf{\epsilon})(i))) \\
	&\equiv C(a, \fun{j}.a). \\
	\end{align*}
	
	\noindent In the $i$-direction, for $i=\zero$ we have 
	
	\begin{align*}
	D_\zero 
	&\equiv \fun{k}. C(\alpha(k)(\zero), \fun{j}.\filler_{j}(\alpha(k)(\zero))) \\
	&\equiv \fun{k}. C(\refl{a}(k \land^* \one), \fun{j}. \refl{a}(k \land^* j)), 
	\end{align*}
	
	\noindent and for $i=\one$ 
	
	\begin{align*}
	D_\one 
	&\equiv \fun{k}. C(\alpha(k)(\one), \fun{j}.\filler_{j}(\alpha(k)(\one))) \\
	&\equiv  \fun{k}. C(a, \fun{j}.a). \\
	\end{align*}
	
\end{proof} 


The proof above is inspired by an argument given by Angiuli~\cite[pp.54--56]{angiuli2019computational}, who considered a slightly different definition of path induction based on the folklore result that $\sum_{(x:A)} \path{A}{a}{x}$ is contractible with center $\langle a, \refl{a} \rangle$ and that type families respect paths in their indexing type. 

In the presence of higher inductive types~\cite{ch,chm}, the inductive path type from homotopy type theory~\cite{hottbook} can be recovered as a higher inductive type freely generated by reflexivity~\cite{ch}. Path induction then acts strictly on reflexivity, since it is given as a specific generator that can be recognized by the eliminator, and this allows for a strict computation rule. This inductive path type can be shown to be equivalent to the path type~\cite{ch}, meaning that, by univalence, the inductive path type is the path type up to a path. 

\subsection{The Eckmann--Hilton argument}

In homotopy type theory, paths from a point to itself are called loops. Thus, given a type $A$ and a point $a : A$, the loop space $\Omega(A,a)$ is defined to be the type of loops $\path{A}{a}{a}$. When the loop space of a loop space is considered, we have the second loop space $\Omega^2(A,a)$, which is the type $\path{\path{A}{a}{a}}{\refl{a}}{\refl{a}}$. 

The Eckmann--Hilton argument is one interesting result about the second loop space that states that path concatenation is commutative. It is inspired by a classical result in homotopy theory and proven by path induction in homotopy type theory~\cite[thm.2.1.6]{hottbook}. We adapt this proof in what follows:

\begin{theorem}[Eckmann--Hilton]
	For any $\alpha, \beta : \Omega^2(A,a)$, there is a path
	$$\path{\Omega^2(A,a)}{\trans{\alpha}{\beta}}{\trans{\beta}{\alpha}}$$
\end{theorem}
\begin{proof}
	
	It is easier to prove a stronger statement that holds more generally for any two-dimensional globular paths by path induction, and then derive the intended claim as a special case. 
	So, given paths $\alpha : \path{\path{A}{a}{b}}{p}{q}$ and $\beta : \path{\path{A}{b}{c}}{r}{s}$, first we define a right whiskering operation 
	
	$${\alpha}\rsq{r} : \path{\path{A}{a}{c}}{\trans{p}{r}}{\trans{q}{r}}$$
	in the obvious way such that the following composition holds, with $\alpha$ at the bottom:
	
	
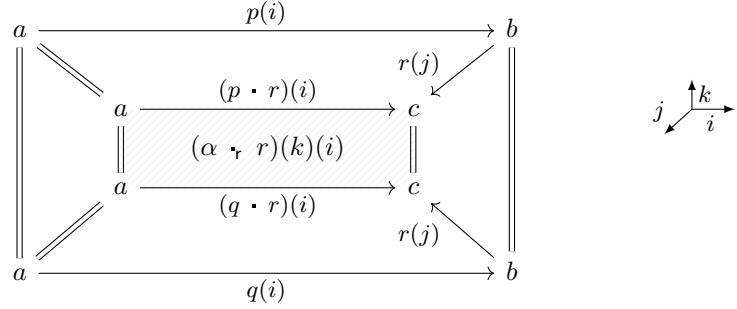
\vspace{2mm}
	
	\begin{tikzcd}
		[execute at end picture={
			\foreach \Valor/\Nombre in
			{
				tikz@f@30-2-4/a,tikz@f@30-2-8/b,tikz@f@30-3-4/c,tikz@f@30-3-8/d%
			}
			{
				\coordinate (\Nombre) at (\Valor);
			}
			\fill[pattern=north east lines,pattern color=grey,opacity=0.3]
			(b) -- (a) -- (c) -- (d) -- cycle;
		}]
		&& a \arrow[dr,equal] \arrow[rrrrrr,"p (i)"] &&&&&& b \arrow[dl,swap,"r (j)" near end] && {} \\
		&&& a \arrow[rrrr,"(\trans{p}{r})(i)"] \arrow[rrrrd,phantom,"({\alpha}\rsq{r}) (k) (i)"] &&&& c &&& \arrow[r,shorten >= 12pt,-latex,swap,start anchor=center,"i" near start] \arrow[u,shorten >= 12pt,-latex,swap,start anchor=center,"k" near start] \arrow[dl,shorten >= 28pt,-latex,swap,start anchor=center,"j" near start] & {}    \\
		&&& a \arrow[rrrr,swap,"(\trans{q}{r})(i)"] \arrow[u,equal] &&&& c \arrow[u,swap,equal] && {} \\
		&& a \arrow[ur,swap,equal] \arrow[uuu,equal] \arrow[rrrrrr,swap,"q (i)"] &&&&&& b
		\arrow[ul,"r (j)" near end] \arrow[uuu,swap,equal]  \\
	\end{tikzcd}
	
	\vspace{2mm}
	
	\noindent then we define left whiskering
	$${p}\lsq{\beta} : \path{\path{A}{a}{c}}{\trans{p}{r}}{\trans{p}{s}}$$
	by composition on a similar open cube but with $\beta$ at the right:
	
	
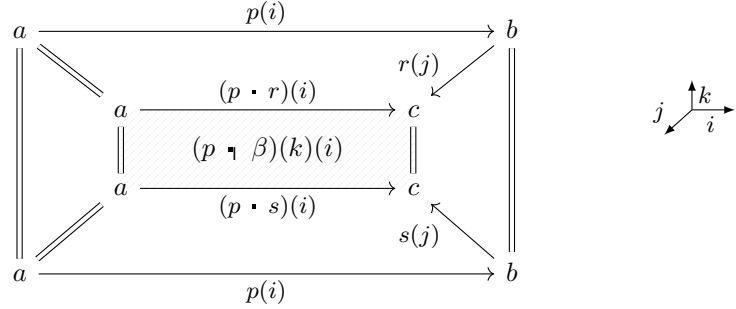
\vspace{2mm}
	
	\begin{tikzcd}
		[execute at end picture={
			\foreach \Valor/\Nombre in
			{
				tikz@f@30-2-4/a,tikz@f@30-2-8/b,tikz@f@30-3-4/c,tikz@f@30-3-8/d%
			}
			{
				\coordinate (\Nombre) at (\Valor);
			}
			\fill[pattern=north east lines,pattern color=grey,opacity=0.3]
			(b) -- (a) -- (c) -- (d) -- cycle;
		}]
		&& a \arrow[dr,equal] \arrow[rrrrrr,"p (i)"] &&&&&& b \arrow[dl,swap,"r (j)" near end] && {} \\
		&&& a \arrow[rrrr,"(\trans{p}{r})(i)"] \arrow[rrrrd,phantom,"({p}\lsq{\beta}) (k) (i)"] &&&& c &&& \arrow[r,shorten >= 12pt,-latex,swap,start anchor=center,"i" near start] \arrow[u,shorten >= 12pt,-latex,swap,start anchor=center,"k" near start] \arrow[dl,shorten >= 28pt,-latex,swap,start anchor=center,"j" near start] & {}    \\
		&&& a \arrow[rrrr,swap,"(\trans{p}{s})(i)"] \arrow[u,equal] &&&& c \arrow[u,swap,equal] && {} \\
		&& a \arrow[ur,swap,equal] \arrow[uuu,equal] \arrow[rrrrrr,swap,"p (i)"] &&&&&& b
		\arrow[ul,"s (j)" near end] \arrow[uuu,swap,equal]  \\
	\end{tikzcd}
	
	\vspace{2mm}
	
	\noindent It is easy to see by path induction on $\alpha$, $\beta$, $p$, and $r$ that there exists a path
	
	\begin{equation*}
	\path{\path{\path{A}{a}{c}}{\trans{p}{r}}{\trans{q}{s}}}{\trans{({\alpha}\sq_{\mathsf{r}}{r})}{({q}\sq_{\mathsf{l}}{\beta})}}{\trans{({p}\sq_{\mathsf{l}}{\beta})}{({\alpha}\sq_{\mathsf{r}}{s})}}.
	\end{equation*}
	
	\noindent Now let $p \equiv q \equiv r \equiv s \equiv \refl{a}$. Since reflexivity is both a right and left unit for path concatenation (as shown in Lemmas \ref{lem:ru} and \ref{lem:lu}), the above proposition demonstrates our intended claim.

\end{proof}

The Eckmann--Hilton argument has been proved very recently for De Morgan cubes in a purely cubical way that avoids path induction~\cite{blm21}. The key element that makes this proof possible is a transportation from $\zero$ to $\one$ along a $j$-type line of paths from $\trans{\alpha_{\mathsf{r}}(j)}{\beta_{\mathsf{l}}(j)}$ to $\trans{\beta_{\mathsf{l}}(j)}{\alpha_{\mathsf{r}}(j)}$ in $\Omega(A,a)$, constructed using function application, path concatenation, and the right and left unit, where

\begin{align*}
\alpha_{\mathsf{r}}(j) :\equiv & \mathsf{ap}_{\fun{p}. \mathsf{ru}_p^{-1}(j)} (\alpha) : \path{\Omega(A,a)}{\mathsf{ru}_{\refl{a}}^{-1}(j)}{\mathsf{ru}_{\refl{a}}^{-1}(j)} \\
\beta_{\mathsf{l}}(j) :\equiv & \mathsf{ap}_{\fun{p}. \mathsf{lu}_p^{-1}(j)} (\beta) : \path{\Omega(A,a)}{\mathsf{lu}_{\refl{a}}^{-1}(j)}{\mathsf{lu}_{\refl{a}}^{-1}(j)}.
\end{align*}

\noindent This type is actually well-formed for De Morgan cubes because it can be shown that $\mathsf{ru}_{\refl{x}} \equiv \mathsf{lu}_{\refl{x}}$, for all $x : A$, using proofs of the right and left unit laws based on essentially the same compositions, except that they make use of the built-in connections that act strongly on reflexivity terms. This definitional equality, however, does not hold for cartesian cubes and their weak connections, so a direct cubical proof would have to rely on a different strategy. 


\section{Directions for future work}

There is much to be done yet in order to provide a cubical alternative to the informal type theory project of the homotopy type theory book~\cite{hottbook}. We view this paper as opening up many possibilities for future work, including informal cubical accounts of the higher groupoid structure of type formers, univalence, higher inductive types, homotopy $n$-types, and the development of mathematics such as homotopy theory, category theory, or set theory.

Part of the proofs contained in this paper have been formalized in the proof assistants \texttt{Cubical Agda}~\cite{vezzosi2019cubical} and \texttt{\color{red}red\color{black}tt}~\cite{redtt}, and are available online.\footnote{
	See e.g. \url{https://github.com/RedPRL/redtt/blob/master/library/prelude/path.red} for some basic constructions for paths in the cartesian style, including function extensionality and function application, path inversion and concatenation, and path induction and its computation rule. For proofs of the groupoid laws using De Morgan cubes, see also \url{https://github.com/agda/cubical/blob/master/Cubical/Foundations/GroupoidLaws.agda}.} 

\theendnotes 


\vspace{5mm} 
\noindent \textbf{Acknowledgments} \; The author wishes to thank Carlo Angiuli, Steve Awodey, Evan Cavallo, Robert Harper, Anders M\"{o}rtberg, and two anonymous referees for helpful comments on an earlier draft of this paper. This work was supported by the US Air Force Office of Scientific Research (AFOSR) grant FA9550-18-1-0120. Any opinions, findings and conclusions or recommendations expressed in this material are those of the author and do not necessarily reflect the views of the AFOSR. 


\bibliography{ref}
\bibliographystyle{plain}

\end{document}